\RequirePackage{amsmath}
\documentclass[cmp,envcountsame]{svjour}

\usepackage[english]{babel} %english language
\usepackage[T1]{fontenc} %256 bit font encoding

\usepackage[colorlinks=true]{hyperref}
\usepackage[margin=3.5cm]{geometry}

\usepackage{times}
\usepackage[table,usenames,dvipsnames]{xcolor}

\usepackage[square,sort&compress,comma,numbers]{natbib}
\usepackage{subfig}

\usepackage{amssymb} %lots of math symbols
\usepackage{mathtools} %for, e.g., \DeclareMathOperator
\usepackage{amsmath} %for spacing in equations and aligns
\usepackage{amsfonts}
\usepackage{dsfont}
\usepackage{ytableau}
\usepackage{braket}
\usepackage{tikz} %tikz picture
\usetikzlibrary{decorations.pathmorphing}
\usetikzlibrary{shadows}
\usepackage{appendix}

\newcommand{\proofend}{\hfill\fbox\\\medskip}

\newtheorem{observation}[theorem]{Observation}
\newcommand{\tr}{\mathop{}\mathopen{}\mathrm{tr}}	% see http://tex.stackexchange.com/questions/16649/too-small-space-when-using-declaremathoperator/16699#16699
\DeclarePairedDelimiter{\abs}{\lvert}{\rvert}

\makeatletter
\newcommand{\newreptheorem}[2]{%
\newenvironment{rep#1}[1]{%
 \def\rep@title{#2 \ref{##1}}%
 \begin{rep@theorem}}%
 {\end{rep@theorem}}}
\makeatother

\newreptheorem{theorem}{Theorem}
\newreptheorem{lemma}{Lemma}

\definecolor{martin}{rgb}{0,.1,1}

\usepackage{cancel}% for \cancel
\usepackage[normalem]{ulem} % for \sout

\definecolor{jens}{RGB}{40,180,40}

\definecolor{Emilio}{RGB}{255,128,0}

\newcommand{\e}{\mathrm{e}}
\newcommand{\cov}{\mathrm{cov}}
\renewcommand{\i}{\mathrm{i}}

\DeclareMathOperator{\Tr}{Tr}
\DeclareMathOperator{\vect}{vec}

\newcommand{\ad}{\mathrm{ad}}

\def\equationautorefname~#1\null{eq.~(#1)\null}%to make autoref better for equations

\newcommand{\EE}{\mathbb{E}}
\newcommand{\CC}{\mathbb{C}}
\newcommand{\E}{\mathbb{E}}
\newcommand{\1}{\mathds{1}}

\newcommand{\RR}{\mathds{R}}

\newcommand{\U}{\mathds{U}}
\newcommand{\SU}{\mathds{SU}}
\newcommand{\mc}[1]{\mathcal{#1}}
\renewcommand{\O}{O}
\renewcommand{\H}{\mc{H}}
\newcommand{\LL}{\mathcal{L}}
\newcommand{\norm}[1]{\Vert #1 \Vert}

\newcommand{\kw}[1]{\frac{1}{#1}}

\newcommand{\Haar}{\mathrm{Haar}}

\newcommand{\SLH}{\mathrm{SLH}}
\DeclareMathOperator{\su}{\mathfrak{su}}
\renewcommand{\u}{\operatorname{\mathfrak{u}}}

\makeatletter % metadata for hyperref
 \hypersetup{pdftitle = {Mixing properties of stochastic quantum Hamiltonians},
	     pdfauthor = {Emilio Onorati,  Oliver Buerschaper, Martin Kliesch, Winton Brown, Albert H. Werner, Jens Eisert},
	     pdfsubject = {},
	     pdfkeywords = {Stochastic Hamiltonian, Brownian motion, Wiener process, unitary group, quantum tensor product expanders, approximate unitary design, moment operator, gap}
	    }
\makeatother

\begin{document}
\title{Mixing properties of stochastic quantum Hamiltonians}
\author{E.\ Onorati$^1$,  O.\ Buerschaper$^1$, M.\ Kliesch$^1$, W.\ Brown$^{1}$, A. H.\ Werner$^{1,2}$,   J.\ Eisert$^1$}
\institute{1 Dahlem Center for Complex Quantum Systems, Freie Universit\"{a}t Berlin, 14195 Berlin, Germany\\
2 Department of Mathematical Sciences, University of Copenhagen, Universitetsparken 5, DK-2100 Copenhagen, Denmark}

\date{}
\vspace*{-4cm}
\begin{minipage}{\linewidth}
\maketitle
\end{minipage}

\begin{abstract}
Random quantum processes play a central role both in the study of fundamental mixing processes in quantum mechanics related to equilibration, thermalisation and fast scrambling by black holes, as well as in quantum process design and quantum information theory. In this work, we present a framework describing the mixing properties of continuous-time unitary evolutions originating from local Hamiltonians having time-fluctuating terms, reflecting a Brownian motion on the unitary group. The induced stochastic time evolution is shown to converge to a unitary design. As a first main result, we present bounds to the mixing time. By developing tools in representation theory, we analytically derive an expression for a local $k$-th moment operator that is entirely independent of $k$, giving rise to approximate unitary $k$-designs and quantum tensor product expanders. As a second main result, we introduce tools for proving bounds on the rate of decoupling from an environment with random quantum processes. By tying the mathematical description closely with the more established one of random quantum circuits, we present a unified picture for analysing local random quantum and classes of Markovian dissipative processes, for which we also discuss applications.
\end{abstract}

{\hypersetup{linkcolor=cyan}
\begin{minipage}{\linewidth}
\tableofcontents
\end{minipage}

%====================================================================================
%====================================================================================
\section{Introduction and motivation}
%====================================================================================
%====================================================================================

%\subsection{Random quantum processes}
In recent years, several ramifications of quantum processes having a random component
have become prominent in the literature \cite{Convergence,Scrambling,Speedups,Decoupling,Oliveria,RandomErrorCorrection,RandomHamiltonian,SuperpolySpeedup,FastScrambling}.
These are quantum mechanical processes, but ones which have a classical random component.  \emph{Random circuits} are discrete processes of this type, so quantum circuits composed of unitary quantum gates, each quantum gate being randomly drawn according to some probability measure \cite{RandomCircuitsLow,Speedups,SuperpolySpeedup,BoutenHandel,BrandaoHarrowHorodecki,Pseudorandom}. Continuous-time processes belong to this class too, e.g.  \emph{stochastically fluctuating local Hamiltonians} \cite{FastScrambling,Belton}.
For both discrete and continuous-time evolutions, similar questions arise. This is, for instance, the question
how long it takes or what depth of the quantum circuit is required until suitable mixing -- in a sense made precise
below -- is achieved, meaning that they approximate a so-called  \emph{unitary design} \cite{UnitaryDesigns}.

%\subsection{Markov chain mixing tools}
This recent development parallels and further develops an established body of literature on fully classical random processes:
 \emph{Markov chain mixing} provides tools to capture details of the convergence of a Markov chain to its stationary distribution, giving tight bounds on relevant time scales for mixing, hitting or cover times \cite{MarkovChains}. Applications of this powerful mathematical framework range from algorithms design in computer science to the understanding of processes of equilibration and thermalisation in classical statistical mechanics. The famous cut-off phenomenon of card-shuffling is epitomic for the many intriguing insights the theory has to offer, showing that decks of $52$ cards have to be shuffled seven times until  the distribution is suddenly close in variation distance to the uniform mixture \cite{ShufflingCards}.

%\subsection{Applications of random quantum processes in quantum information theory}
Such  \emph{random quantum processes}, as they will be called in this work, again have applications in algorithms design, now  \emph{quantum algorithms design} \cite{Speedups,SuperpolySpeedup,BriegelMixing}. They are used in  \emph{quantum process tomography}
and  \emph{low rank matrix recovery} \cite{tomographydesigns,Efficient} and  \emph{benchmarking} \cite{Benchmarking}, where they provide powerful tools
to avoid significant overheads otherwise necessary with naive deterministic prescriptions.
They play a key role in notions of  \emph{decoupling} \cite{Decoupling}, the task of approximately bringing a
quantum mechanical system into a tensor product state with its environment, which constitutes a key property of
quantum mixing processes. It also is an important primitive in quantum information theory \cite{SzeDuTomRen13,Decoupling}:
Indeed, it plays a central role in arguments of state merging \cite{Merging},  the task of conveying a subsystem from a sender to a receiver.
It also is key to the proof of the quantum reverse Shannon Theorem \cite{ReverseShannon} and is useful to capture
quantum channel capacities \cite{Buscemi}.
 \emph{Error correcting codes}, so codes that protect quantum information against unwanted local decoherence, can be built upon such random processes \cite{RandomErrorCorrection}. It should be clear from this that the analysis of such
processes constitutes a powerful proof tool in the context of quantum information theory.

%\subsection{Random quantum processes as proxies for equilibration and black hole thermalisation}
Maybe most intriguingly, they are used as proxies for natural mixing processes occurring in physical systems
governed by quantum mechanical laws. Clearly, random processes are reminiscent in many ways and sometimes
exactly model  \emph{thermalising dynamics} of interacting quantum systems with many constituents \cite{bigreview}.
This link
has particularly prominently been explored in the context of  \emph{black hole thermalisation}.
This phenomenon is connected to the still unresolved puzzle how quickly black holes release information
about their microscopic state. Based on considerations from string theory
and gauge-gravity correspondences \cite{Banks,Maldacena},
it is increasingly becoming clear that black holes do not destroy information when evaporating. This
insight raises the question on what time scales this release of information precisely happens.
It has been suggested that the time scale is set by the time it takes to  \emph{``scramble''}
the microscopic degrees of freedom of the black hole, in a way that initial local perturbations will
be locally undetectable. Taking this idea seriously, it has been suggested in the famous \emph{ ``fast scrambling conjecture''} that black holes should indeed be perfect scramblers,
taking a time logarithmic in the number of degrees of freedom \cite{FastScrambling,FastScramblingConjecture}. Unfortunately, the microscopic models under consideration, most importantly the so-called
 \emph{matrix models} \cite{Banks}, involve highly non-local interactions in interacting models that embody both
bosonic and fermionic degrees of freedom, are notoriously difficult to solve, even on modern supercomputers.
For this purpose,
research on the fast scrambling conjecture has focused much on identifying proxies that share
many similarities with the actual physical model, to get a handle on a precise quantitative
understanding of the mechanisms that lead to such a fast scrambling. Classical models
have been considered \cite{Berenstein}, small-dimensional quantum models
\cite{Huebener}, random processes precisely of the kind considered here,
in the form of stochastically fluctuating local Hamiltonians \cite{FastScrambling},
as well as random circuits \cite{Scrambling}. It is one of the key motivations of the present work to provide tools for studies of this kind. Indeed,
stochastically  fluctuating local Hamiltonians are less well understood than
random circuits-- and importantly, a precise understanding
of the equivalences of mixing times seems urgently needed.

%\subsection{Time-fluctuating processes}
It should also be clear that  \emph{time-fluctuating processes} as such are ubiquitous in nature. For such
processes, the dynamics is captured by a family of Hamiltonians of the form
\begin{equation}\label{Hami}
    H_t = H_0 + F_t,
\end{equation}
such that both $H_0$ and $F_t$ are local Hamiltonians of a quantum system with many degree of freedom and $F_t$ is
randomly fluctuating in time.
Any
experimental setting in quantum mechanics will necessarily be interacting with a classical exterior in one way or the other.
Many decoherence mechanisms can well be approximated by a classical degree of freedom
fluctuating randomly in time. In fact, effects like magnetic field fluctuations are of this type, and so are
Gaussian noisy processes in condensed matter physics. This type of noise is usually
seen as a detrimental type of  \emph{decoherence}, deteriorating the coherence present in the
quantum mechanical system. This connection to  \emph{local dissipative dynamics} will be made clear below.
Again, a precise understanding of these effects and their impact seems
desirable.

Again more technologically or pragmatically 
speaking, it should be clear that fluctuating Hamiltonians of the form (\ref{Hami}) by no means have to reflect
unwanted external noise. Quite to the contrary, in many applications in which random quantum circuits are envisioned, one can as well
replace the quantum circuit by the mere time evolution under such a fluctuating Hamiltonian. In many situations this can lead to a significantly
simplified prescription, compared to implementing precisely controlled quantum gates that are designed according to samples of 
some suitable classical probability distribution. That is to say, in a number of instances, fluctuating Hamiltonians can be seen as being 
vastly more feasible than random circuits that require the accurate realization of quantum gates, 
at least from the perspective of  implementation. 

%\subsection{Overview of this work}
Motivated by these considerations, 
in this work we investigate
mixing properties of random quantum processes in quantum many-body systems.
Specifically, we consider a family of time-fluctuating local Hamiltonians inducing a Brownian motion on the unitary group.
We show that this locally generated Brownian motion gives rise to an efficient approximate unitary $k$-design of arbitrary order, i.e. all
its moment operators converge to those of the Haar measure. Furthermore, the convergence rate is comparable to that of a random quantum circuit in discrete time. Our main technical contribution is a connection between the generator of the local diffusion and the Casimir element of the special unitary group. This allows us in turn to obtain an explicit uniform lower bound on the gap of the local generator, i.e. independent of the order $k$. Hence, our results also provide a class of probability measures on the unitary group, where the set of generated unitaries has the spectral gap property, albeit with an explicitly known constant \cite{bourgain2012spectral, benoistspectral}. This might be an unexpected result, as the convergence time of the $k$-th moment increases with $k$ for many processes.

We also show decoupling with almost linear scaling in the system size. We interpret the time-fluctuating Hamiltonian in the framework of a  \emph{continuous-time random walk}, relating it with the discrete random walk induced by random quantum circuits with Haar distribution. The continuous-time version has been first formalised by Montroll and Weiss \cite{MontWeiss} as a sequence of random transitions (jumps) spaced out by waiting times and been object of successive study
\cite{Weiss}, being applied to a wide range of fields of physics \cite{ZabDenHan,SchuBar,ChauKob}.
In particular, an exact correspondence between the accelerated steps of the random walk induced by random quantum circuits given in  ref.~\cite{RandomCircuitsLow} and the jumps of the continuous-time random walk generated by the fluctuating Hamiltonian infers a close similarity between the discrete circuit and the continuous process and can be hence used to relate results from these two settings.
Much of the present work can hence be seen as providing a unifying framework to capture random quantum processes---continuous and discrete in time---under a single umbrella. By bringing notions of fluctuating Hamiltonians
closer together with those of random quantum circuits, we provide a unified picture of mixing properties of random quantum processes. The results laid
out here are expected to provide powerful technical tools to make further progress in those research questions for which such
quantum processes having a classical random component are relevant, maybe most intriguingly the fast scrambling conjecture.

%==================================================
%==================================================
\section{Preliminaries}
%==================================================
%==================================================
In this section, we introduce basic notions and concepts that will be made use of when stating the main results.
In the focus of attention will be the concept of a unitary design.
A unitary design is a probability distribution over unitary matrices that mimics properties of the Haar measure, in a similar sense as a
spherical design approximates the unit sphere.
In order to capture approximate versions of unitary designs -- which will feature strongly throughout this work -- several norms will be relevant.
We finish this section by introducing the concept of locally generated Brownian motion on the unitary group.

%===================================================
\subsection{Norms}
%===================================================
In this work, all Hilbert spaces and other vector spaces considered are finite dimensional, reflecting finite-dimensional quantum systems.
We denote the vector space of linear operators on a Hilbert space $\H$ by $\LL(\H)$.
The \emph{trace norm} of an operator $X\in\LL(\H)$ is defined by
\begin{equation}
\left\|X\right\|_1\coloneqq\sum_{i} \, s_i(X) ,
\end{equation}
where we denote by $s_i(X)$ the $i$-th singular value of $X$.

For linear maps on operators we introduce the following two norms.
% each having an operational meaning. mk: what is the operational meaning of the infty-norm?
\begin{enumerate}
\item The  \emph{diamond norm} \cite{Watrous1} of $\mathcal{T}\in \LL(\LL(\H))$ is defined to be
\begin{equation}\label{defdiamond}
\left\| \mathcal{T} \right\|_{\diamond} \coloneqq \sup_{d} \sup_{X\neq 0}\, \frac{\left\|\left(\mathcal{T}\otimes \mathcal{I}_{d}\right) X \right \|_1}{\left\|X\right\|_1}
\, ,
\end{equation}
where $\mathcal{I}_{d}$ denotes the identity element in $\LL(\LL(\CC^d))$.
The diamond norm is most meaningful to quantify how close two
quantum channels are, reflecting physical processes.
\item The  \emph{infinity norm} $\| \mathcal{T} \|_\infty$ of $\mathcal{T}$ is given as the largest singular value of $\mc T$.
\end{enumerate}
%These norms will feature when defining the concept of an \emph{approximate unitary design}.

%======================================================
\subsection{Exact and approximate unitary designs}
%======================================================
In this work, we examine probability measures over the unitary group and their mixing properties. We refer to them as distributions on the unitary group $\U(N)$.
% $N$ taking several meanings here, usually the dimension $d^n$ of a quantum system of $n$ constituents of local dimension $d$.
A central role is played by the invariant distribution over all unitaries given by the  \emph{Haar measure}.
In many applications, one is interested in distributions which approximate properties of the Haar measure but can be generated with limited resources (e.g.
quantum circuits of a certain given depth or
Hamiltonian quantum evolution generated by local Hamiltonians for a certain run-time). \emph{Unitary $k$-designs} capture the ability of a distribution to mimic the properties of the Haar measure in the sense that expectation values of polynomials of a certain order $k$ are equal to those of the Haar measure.
As pointed out above, they have a wide range of applications in quantum algorithm design
\cite{Speedups,SuperpolySpeedup,BriegelMixing}, in quantum state and process tomography
\cite{tomographydesigns,Efficient}, and in notions of benchmarking \cite{Benchmarking} -- basically as a powerful tool for \emph{partial de-randomisation}.
Conceptually,
they feature strongly in descriptions of equilibration, thermalisation and scrambling  \cite{FastScrambling,FastScramblingConjecture,Scrambling}.

In order to make this concept of a unitary design precise, we define the \emph{$k$-th moment operator} $M_\mu^k$
on $\mathcal{L}(\mathcal{H}^{\otimes k})$ with respect to a distribution $\mu$ on $\U (N)$ by
\begin{equation}
	\label{eq:k-moment}
	X\mapsto M_{\mu}^{k}(X) \coloneqq \E_{\mu} \left[ U^{\otimes k}X\mkern2mu(U^{\dagger})^{\otimes k} \right] \, .
\end{equation}
Exact unitary designs as well as suitable approximate versions thereof can be defined in terms of $M_\mu^k$ .

\begin{definition}[Unitary designs]\label{defdiamondnorm}
	Let $\mu$ be a distribution over the unitary group $\U (N)$. Then $\mu$ is an  \emph{$\varepsilon$-approximate unitary $k$-design} if
	\begin{equation}
	\left\|  M_{\mu}^{k} - M_{\Haar}^{k}\right\|_{\diamond} \leq \varepsilon .
	\end{equation}
	For $\varepsilon=0$, the distribution $\mu$ is also called an  \emph{exact unitary $k$-design}.
\end{definition}

Physically implementing an exact unitary design is in general neither an obvious nor an efficient task.
Fortunately, for a plethora of applications, exactness of a design is not required.
Instead, we are usually interested in obtaining  \emph{approximate unitary designs}, i.e., distributions which behave similarly as the Haar measure and which can be implemented efficiently. We would also like to note that there are different formal definitions of approximate unitary designs, each of which being equipped
with a different interpretation and being relevant in a different context; a
close examination has been done in ref.\ \cite{PhDLow}.

An important method for obtaining a bound on $\varepsilon$ is to analyse the gap of the moment operator~$M_\mu^k$, leading to the following definition.

\begin{definition}[Tensor product expanders]
	\label{def:TPE}
	A distribution $\mu$ on the unitary group $\U (N)$ is a quantum $(\lambda,k)$-tensor product expander if
	\begin{equation}
	\left\| M_\mu^k - M_\Haar^k \right\|_{\infty} \leq \lambda.
	\end{equation}
\end{definition}

The following lemma links this definition to the one of designs.

\begin{lemma}[{Criterion for being an approximate unitary design \cite[Lemma 2.2.14]{PhDLow}}]\label{lemma:TPE-diamond_relation}
	Let $\mu$ be a distribution on $\U (N)$.
	If $\mu$ is a quantum $(\lambda,k)$-tensor product expander, then $\mu$ is also an $\varepsilon$-approximate $k$-design with $\varepsilon=N^{k}\lambda$.
\end{lemma}

The expectation value of polynomials with respect to the Haar measure over the unitary group can be understood in terms of the
\emph{Schur-Weyl duality}. This says that $M_{\Haar}^{k}$ is an orthogonal projection onto the span of operators representing a permutation
of the $k$~tensor copies of~$\mathcal{H}$ (see ref.~\cite[Proposition 2.2]{SWDuality} for a complete description). This means that
all elements of this space are eigenvectors with unit eigenvalues, while the complement space belongs to the kernel.

A similar connection between eigenvalues and eigenspaces can be made for
\emph{universal distributions} as we will see in Lemma \ref{lemma:k-copy_gapped} below.
In order to amplify closeness of a distribution $\mu$ on $\U (N)$ to the Haar measure,
one can convolute it $\ell$ times with itself and obtain a new measure $\mu^{\star \ell}$
on $\U (N)$. One can effectively draw a unitary $U$ from $\mu^{\star \ell}$ by drawing
$\ell$ unitaries $U_1, U_2, \dots, U_\ell$ independently from $\mu$ and take $U$ as the
product $U = U_1 U_2 \dots U_\ell$. Importantly, it holds that
\begin{equation}
M^k_{\mu^{\star \ell}} = (M^k_{\mu})^\ell \, .
\end{equation}

If the support of $\mu^{\star \ell}$ becomes dense in $\U (N)$ for large $\ell$ we call $\mu$  \emph{universal}. More precisely, a universal distribution can be defined as follows.

\begin{definition}[Universal distribution]
	Let $\mu$ be a distribution on $\U (N)$. Then $\mu$ is said to be  \emph{universal} if for all $V \in U (N)$ and any $\delta >0$
	there exists a positive integer $\ell$ such that
	\begin{equation}
	\mu^{\star \ell}\left(B_\delta(V)\right)>0,
	\end{equation}
	where $B_\delta(V)$ is the neighbourhood of $V$ with radius $\delta>0$.
\end{definition}
Here, the canonical way to capture the radius is in terms of the geodesic distance on $\U (N)$. It should be clear,
however, that any other equivalent metric gives rise to the same definition of universality.
This definition can be seen as a generalisation of a universal gate set: if $\mu$ is the uniform distribution over finitely many unitaries then this set of unitaries is universal if and only if $\mu$ is universal.
Universal distributions induce moment operators satsifying the following property for all orders $k$.

\begin{lemma}[Lemma 3.7 in ref.~\cite{RandomCircuitsLow}]\label{lemma:k-copy_gapped}
	Let $\mu$ be a distribution on $\U  (N)$. Then all eigenvectors of~$M_\Haar^k$ with unit eigenvalue
	are eigenvectors of~$M_\mu^k$ with unit eigenvalue. Additionally, if $\mu$ is universal then $\mu$ is $k$-copy gapped
	for any positive integer $k$. This means that
	\begin{equation}
	\left\| M^k_\mu - M^k_\Haar \right\|_\infty < 1.
	\end{equation}
\end{lemma}

As a consequence, $M^k_{\mu^{\star \ell}}$ converges to $M^k_\Haar$ for $\ell \to \infty$.
For many practical applications, however, a bound on the convergence rate is needed. Below, we will extend such a bound from quantum circuits \cite{BrandaoHarrowHorodecki} to locally generated Brownian motion on $\U(N)$.

%======================================================
\subsection{Locally generated Brownian motion on the unitary group}\label{BM_on_U}
%======================================================
In this section, we define the central objects studied in this work: Brownian motions on the unitary group. These are continuous-time stochastic processes describing the unitary evolution of a quantum system with a fluctuating Hamiltonian and whose increments satisfy specific properties. In this way, a distribution on $\U(N)$ is induced, which changes over time and eventually converges to the Haar measure for arbitrary moments.

A deterministic family of Hamiltonians $t\mapsto H_t$ depending continuously on time $t\in \RR$
generates a unitary time evolution (see refs.~\cite{Dyson,Dollard}) via the time-ordered exponential
\begin{equation}
U_t = \mc T\left[ \exp \left\{ -\i\,\int_0^t H_s\, ds \right\}\right]
\, .
\end{equation}
% We will write $H_t$ for the entire family of Hamiltonians from now on.
Conversely, $H_t$ can be recovered from the  \emph{increments}
\begin{equation}
  U_t^\dagger U_{t+\Delta t}
  =
  \1 -\i \Delta t \, \langle H_t \rangle + O(\Delta t^2)
\end{equation}
by taking the limit $\Delta t \to 0$, where
$\langle H_t \rangle$ denotes the time average of $H_s$ over the interval $[t,t+\Delta t)$.

In this work, we investigate how well a Brownian motion $U_t$ that has stochastic
increments of the form of a local fluctuating Hamiltonian $H_t$ generates an approximate
unitary $k$-design as a function of time $t$.
First of all, we define Brownian motion as follows
(c.f. refs.~\cite{ItoBM,Unitary_BM,Lia04}):
\begin{definition}[Brownian motion on the unitary group]
	\label{def:UWP}
	A process $U_t$ on the unitary group $\U(N)$ is called  \emph{Brownian motion} if the following conditions are satisfied.
	\begin{enumerate}
		%[label=(\subscript{A}{\arabic*})]
		\item For all $0 < t_1 < t_2 < \dots < t_n$, the (left) increments
		$U_{t_1}U_0^\dagger, U_{t_2} U_{t_1}^\dagger, \dots,
		U_{t_{n}} U_{t_{n-1}}^\dagger$ are independent.
		\item For any time~$t\geq 0$, the increments are stationary, i.e., for any $\Delta t>0$ the increment $ U_{t+\Delta t} U_{t}^\dagger$ is equal in distribution to $U_{\Delta t}U_0^\dagger$.
		\item The paths $t\mapsto U_t$ are continuous almost surely.
	\end{enumerate}
\end{definition}
Brownian motion $U_t$ on the Lie group $\U(N)$ corresponds to Brownian motion $W_t$ on the Lie algebra $\u(N)$ through the exponential map,
which for a matrix Lie group is given by the series $\exp(X)=\sum_{n=0}^\infty {X^n}/{n!}$.
%(see refs.~\cite{Unitary_BM,BM_Diffeom_group,Ito_exponential}).
More precisely, one can construct Brownian motion on $\U(N)$ by  \emph{injecting the differential} of a Brownian motion from $\u(N)$ via the product integral of the exponential map,
\begin{equation}\label{eq:def:U_t}
U_t= \lim_{\Delta t \rightarrow 0} \prod_{\ell=t/\Delta t}^{1} \exp \left\{ W_{\ell \Delta t} - W_{(\ell-1)\Delta t} \right\} \, U_0 \, ,
\end{equation}
see ref.~\cite[Chapter 4.8]{McKean} for a complete proof of existence and uniqueness of the process, and additionally ref.~\cite[p.\ 226]{Rogers}.

The \emph{Hamiltonian increments} are denoted by
\begin{equation}
	H_{\ell,\Delta t} \coloneqq \i\, \Theta_{\ell, \Delta t}
\end{equation}
with
		\begin{equation}
		\label{eq:pseudo_derivative}
		\Theta_{\ell, \Delta t} \coloneqq \frac{1}{\Delta t} \,\left[ W_{\ell\Delta t} - W_{(\ell-1)\Delta t} \right]
		\end{equation}
being the increments in the Lie algebra $\u(N)$.

\subsection{Local Hamiltonian increments}\label{sec:loc_increments}
We now turn to describing the \emph{local} Hamiltonian increments on the physical quantum system consisting of $n$ subsystems of dimension $d$, so that
$N$ becomes $d^n$.
Those subsystems interact according to an interaction pattern captured by an interaction graph with vertex set $V$ and edge set $E$.
In the special case of $d=2$, this is referred to the \emph{qubit case}, and the system is an \emph{$n$-qubit system}.

We assume that
$\Theta_{\ell, \Delta t}$ from eq.~\eqref{eq:pseudo_derivative} is local with respect to an interaction graph
$(V,E)$, where each vertex in $V$ corresponds to a $d$-level subsystem.
Only qudits connected by an edge $e\in E$ may interact, i.e.,
\begin{equation}
\label{eq:stochastic_local_H} % Now 2 equations!
\Theta_{\ell, \Delta t}= \sum_{e\in E} \theta^{(e)}_{\ell,\Delta t} \, ,
\end{equation}
where each local term $\theta^{(e)}_{\ell,\Delta t}$ is supported on~$e$.
The local terms are explicitly given by
\begin{equation}
\label{eq:stochastic_local_H_details}
\theta^{(e)}_{\ell,\Delta t}
=
-\i \,h_0^{(e)}\,  +  \sum_{\mu} A_\mu^{(e)}\, \xi^{(e,\mu)}_{\ell,\Delta t} ,
\end{equation}
where we specify each term in this equation in the following.
$h_0^{(e)}$ are deterministic Hermitian operators reflecting a constant drift in the evolution.
Each noise operator $A_\mu^{(e)}$ acts on the two vertices connected by $e$ as $A_\mu$ and as the identity elsewhere. $\set{A_\mu}_\mu$ is a basis of the real
Lie algebra
\begin{equation}
%\mathfrak{g}_e \cong
    \u(d^2) \coloneqq \{X \in \CC^{d^2\times d^2}:  \ X=-X^\dagger \}.
\end{equation}
$\xi_k^{(e,\mu)}$ are real random variables representing the noise.
We assume that the noise satisfies
\begin{align}
\label{eq:noise_mean_zero}
 \EE\left[\xi_{\ell,\Delta t}^{(e,\mu)}\right]
 & =0,
%    \label{eq:noise_mean}
    \\
 \EE\left[\xi_{\ell,\Delta t}^{(e,\mu)}\, \xi_{{\ell',\Delta t}}^{(e',\mu')}\right]
 & = - \frac{a}{\Delta t}\, \delta_{\ell,\ell'}\, \delta_{e,e'}\, \kappa^{-1}_{\mu,\mu'},
    \label{eq:noise_cov}
 \end{align}
%  \MK{has cov been defined?}
 where $a>0$ is an arbitrary constant and the matrix $\kappa$ is defined by
\begin{equation}
    \label{eq:Killing_tensor_defining_rep}
    \kappa_{\mu,
            \nu}
    \coloneqq
     -2
     d^2
     \Tr(A_\mu^\dagger
         A_\nu).
\end{equation}
As we will explain later, this matrix is in fact the \emph{Killing metric tensor} associated with the basis~$\set{A_\mu}_\mu$.

\begin{remark}[Orthonormal basis]
    \label{rmk:white_noise}
    If the basis $\set{A_\mu}_\mu$ is orthonormal then our assumption~\eqref{eq:noise_cov} on the covariance simplifies to
    \begin{equation}
        \label{eq:white_noise_cov}
        \EE\left[\xi_{\ell,\Delta t}^{(e,\mu)},\xi_{\ell',\Delta t}^{(e',\mu')}\right]
        =
        \frac{a}{2d^2\Delta t}\,%\mkern2mu
         \delta_{\ell,\ell'}\,
         \delta_{e,e'}\,
         \delta_{\mu,\mu'}
    \end{equation}
    which represents \emph{white noise}.
    This happens, for instance, if we choose the Pauli basis
    (see the example processes in Section~\ref{xmp:white_noise} and Section~\ref{sec:main-decoupling}).
\end{remark}

\begin{remark}[Overcomplete sets of operators]
    \label{rmk:basis_of_algebra}
	Additionally, we may consider an overcomplete set of operators $\set{A_\mu}\in \u(d^2)$ as long as they give rise to a negative contribution to the generator~\eqref{eq:loc_g}, since this will increase the gap of the moment operator induced by the stochastic evolution and hence make the convergence even faster.
\end{remark}

The above described Brownian motion with Hamiltonian increments as in
eq.~\eqref{eq:stochastic_local_H} with the specified $\theta^{(e)}_\ell$ induces a Brownian motion $U_t$ on the unitary group.
We denote the distribution of $U_t$ at time $t$ by $\SLH (t)$ and write the according expectation as $\E_{\SLH (t)}$ .

%====================================================================================
%====================================================================================
\section{Main results}
%====================================================================================
%====================================================================================

In this section, we present the two main results on mixing properties for local stochastic Hamiltonian evolutions. Theorem~\ref{thm_TPE}, together with Corollary \ref{cor_AUD}, asserts that the distribution over unitaries induced by the Brownian motion as in Definition
\ref{def:UWP} is a quantum tensor product expander and hence an approximate unitary $k$-design after a run time scaling polynomially in $k$ and linearly in the system size $n$. This means that it is suitable to efficiently reproduce certain properties of the Haar measure.
Theorem~\ref{Decoupling_main} on decoupling says that any subsystem affected by a stochastic evolution reproducing Brownian motion becomes uncorrelated with respect to a second (possibly initially correlated) subsystem in almost linear run time in system size.

\subsection{Tensor product expanders and approximate unitary designs}
For any of the time-fluctuating local Hamiltonians fulfilling the description in Section~\ref{BM_on_U}, each moment of the generated unitary process becomes close to that of the Haar measure after a sufficiently long run time. This result
can be expressed in terms of quantum tensor product expanders or approximate unitary designs as follows.

\begin{theorem}[Local Brownian motions on $\U (d^n)$ are quantum $(\lambda,k)$-tensor product expanders]
\label{thm_TPE}
Let $U_T$ be a unitary Brownian motions with the increments~\eqref{eq:stochastic_local_H}
	with the interaction graph $(V,E)$ being either a complete
	graph or a $1D$ nearest neighbour graph.
	Then, for any run time
	\begin{equation}
	T \geq 850\lceil\log_d(4k)\rceil^2 d^2 k^5 k^{3.1/ \ln(d)}\, \frac{\ln(1/\lambda)}{a},
	\end{equation}
	$U_T$ is a quantum $(\lambda,k)$-tensor product expander.
\end{theorem}

%  s\coloneqq \left(425\lceil\log_d(4k)\rceil^2 d^2 k^5 k^{3.1/ \log(d)} \right)^{-1}

Then, using Lemma \ref{lemma:TPE-diamond_relation}, we immediately obtain the subsequent corollary.

\begin{corollary}[Approximate unitary $k$-designs]\label{cor_AUD}
	For any run time
	\begin{equation}
	T \geq 850\lceil\log_d(4k)\rceil^2 d^2 k^5 k^{3.1/ \ln(d)}\,
	\frac{nk\log(d) +\ln\left(1/\varepsilon\right)}{a},
	\end{equation}
	$U_T$ is an $\varepsilon$-approximate unitary $k$-design.
\end{corollary}

Theorem \ref{thm_TPE} can be seen as a unifying statement on random quantum processes. It extends the results on random local quantum circuits, as considered in ref.~\cite[Corollary 7]{BrandaoHarrowHorodecki}, to continuous time dynamics under fluctuating Hamiltonians. 
% what is the gap that is bridged? I'd prefer to keep it more concrete. -mk
Note that the scaling of the minimal runtime required for the generating of a unitary $k$-design is by a factor of $n$ smaller with respect to the circuit setting. This is due to the number of Hamiltonian interactions per time step growing linearly in the system size for a 1D graph (which is, as discussed in Lemma \ref{lem:CompleteGraph}, the slowest setting among all complete graphs), while for local random quantum circuits only one gate per step is applied. 
If we re-scale the stochastic Hamiltonian with a pre-factor $O(1/\sqrt{n})$ so that the $k$-th moment operator may be written in the same form as the one induced by a random quantum circuit, i.e. (cf. eqs.~\eqref{generator_global_moment}-\eqref{eq:derive_g_2} for derivation),
\begin{equation}
M^{k}
=
\frac{1}{n} \sum_{j} \left(m^{k}\right)^{j,j+1} ,
\end{equation}
where $\left(m^{k}\right)^{j,j+1}$ denotes the local moment operator applied on qubits $j$ and $j+1$, then we would instead obtain the same scaling for the mixing time. Therefore, we can consider the two scenarios as perfectly compatible. 
% and, without additional effort, [mk: it takes us 35 pages to do that (quite some effort)]
%one can use time-fluctuating random Hamiltonians and random quantum circuits interchangeably for % the generation Haar-random unitary operators. [mk: it is already clear enough. it is also not quite interchangable...] 
As already discussed in the introduction, the time-continuous case might often be the more natural and easier implementable choice in applications, such as tomography or random benchmarking: They do not require the explicit implementation of random bipartite unitary matrices and often the natural fluctuations present in a system are already a good approximation of a locally generated Brownian motion.\\

The full proof of Theorem~\ref{thm_TPE} is given in Section~\ref{section:designs}. Its outline is the following.
\paragraph{Proof idea of Theorem~\ref{thm_TPE}}
  We express the $k$-th moment operator $M^k_{\SLH (t)}$ in terms of a generator $G_k$ so that $M^k_{\SLH (t)}= \exp(tG_k)$. This generator has again the same locality structure as Hamiltonian increments in eq.~\eqref{eq:stochastic_local_H}, i.e., 
  \begin{equation}
  G_k = \sum_{e\in E} g_k^{(e)}.
  \end{equation}
  The crucial point of our proof is to obtain a lower bound to the spectral gap for the invariant subspace of $g_k^{(e)}$.
  This gap can be recovered thanks to a simple relation between $g_k^{(e)}$ and the so-called \emph{Casimir element} of $\su(d^2)$ in a certain (reducible) representation denoted by $\pi_{k,k}$.
  Since  the eigenvalues of the Casimir element in each irreducible representation of~$\su(N)$ are well known, we can determine the spectrum in the representation~$\pi_{k,k}$ from its irreducible decomposition. In particular, using an argument based on the shape of Young diagrams, we derive a local gap, $\Delta(g_k)$, independently of~$k$ from which we finally deduce, applying  results from ref.~\cite{BrandaoHarrowHorodecki}, a gap for $G_k$. This establishes the bound on the mixing time.
\proofend

\subsection{Fast decoupling with stochastic quantum Hamiltonians}
\label{sec:main-decoupling}

The second result is concerned with fast decoupling properties of the
random evolutions considered in this work. In the following, we restrict our analysis to the qubit case in which $d=2$ and where the noise operators are given by the Pauli matrices $\sigma_0 = \1, \sigma_1, \sigma_2,$ and $\sigma_3$.
% \MK{@Jens: this is not entirely for the simplicity of notation!}
We connect to and extend the result on fast decoupling given in ref.~\cite{Decoupling} for a continuous-time evolution.
Inspired by the Hamiltonian given in ref.~\cite{FastScrambling}, for which a reading of the \emph{fast scrambling conjecture} has been studied, we set the increments in eqs.~\eqref{eq:stochastic_local_H} and~\eqref{eq:stochastic_local_H_details} of the Brownian motion to be
\begin{equation}\label{eq:H_Delta_t}
\Theta_{n,\ell,\Delta t}
=
-\i\, H_{n,\ell,\Delta t}
\coloneqq
- \mathrm{i}\, \left(\frac{2}{n(n-1)}\right)^{1/2}\sum_{j< k}\sum_{\alpha, \beta=0}^{3} {\sigma^{j}_{\alpha} \otimes \sigma^{k}_{\beta}\ \xi^{(j, k, \alpha, \beta)}_{\ell,\Delta t}} \, ,
\end{equation}
where $\sigma^{j}_{\alpha} \otimes \sigma^{k}_{\beta}$ means that $\sigma_\alpha \otimes \sigma_\beta$ is applied on qubits labeled $j$ and $k$,
respectively. We recall that $\xi^{(j, k, \alpha, \beta)}_{\ell,\Delta t}$  are i.i.d.\ real random variables with zero mean and covariance
\begin{equation}\label{eq:covariance_xi}
\EE \left[ \xi^{(j, k, \alpha, \beta)}_{\ell,\Delta t} \, \xi^{(j', k', \alpha', \beta')}_{\ell',\Delta t}\right]
=
\frac{1}{\Delta t} \,\delta_{\ell,\ell'} \delta_{k,k'}  \delta_{j,j'} \delta_{\alpha,\alpha'} \delta_{\beta,\beta'}  \  \qquad \forall j,k,\alpha,\beta \,
\end{equation}
which is obtained from eq.~\eqref{eq:white_noise_cov} by choosing $a$.
The pre-factor of $\left({2}/({n(n-1)})\right)^{1/2}$ is chosen so that the initial rate of diffusion of a local operator scales as $O(1/n)$.  This is to normalize the time scale for the diffusion process in order to compare it with the random quantum circuit model in refs.\ \cite{RandomCircuitsLow, Decoupling}, where the probability that a local operator experiences a random gate is $2/n$ per discrete time step.

In what follows, we refer to an $n$-fold tensor product of Pauli operators (including the identity) on qubits as a \emph{Pauli string} and denote it by $\sigma_\alpha$, with an $n$-dimensional index $\alpha= \{0,1,2,3\}^n$ representing the label of each sub-element.

For our results, we need to define the permutation invariance property. This condition is required to deduce a dominant probability distribution on the final Pauli coefficients when starting with an analysis of the evolution of the Pauli weights. Indeed, the random walk on Pauli weights  does not distinguish among strings with same support size but different support, hence it provides the probability distribution for each set of strings with the same support size, but not on Pauli strings taken singularly.

The permutation invariance property has already been debated in the proof of ref.~\cite{RandomCircuitsLow} showing that random quantum circuits with Haar measure are approximate unitary 2-designs. In ref.~\cite{shuffle2design} it has been discussed that this essential condition in the proof had not been granted and an argument making use on random transpositions based on the work of
Diaconis (see refs.~\cite{GenRT,ShufflingCards}) has been put forward solving this issue.
Since we cannot prove that permutation invariance is achieved with sufficiently high probability by the stochastic Hamiltonian evolution itself within a run time scaling almost linearly in $n$, we impose it as a pre-condition for the initial state.
Actually, we can relax the condition and ask for a ``large portion'' of the qubits,  but not necessarily all, to be invariant with respect to an arbitrary permutation. This allows us to apply our result to a larger family of states, for instance those whose support is very small. More formally, we define the permutation invariance property as follows.

\begin{definition}[Permutation invariance property]\label{3P}
	Let $0 \leq \gamma<1$. Let $\sigma_{\pi(\mu)}$ denote a Pauli string whose label is given by interchanging the sub-indices of $\mu$ according to the permutation $\pi$. Then, for an arbitrary quantum state $\rho$ of a $n$-qubit system, we say that it satisfies the  \emph{$\gamma$-permutation invariance property} if there exists a subset of $(1-\gamma)n$ qubits which is invariant with respect to any permutation, i.e., 
	\begin{equation}
	\Tr[\sigma_\mu \rho]= \Tr[\sigma_{\pi(\mu)} \rho]
	\end{equation}
	for every Pauli string $\sigma_\mu$ and every permutation $\pi$ on this subset of qubits.
\end{definition}	
Note that any state $\rho$ with $\left|{\rm supp}(\rho) \right| \leq \gamma n$ is permutation invariant with respect to this definition.

For the decoupling theorem, two initially correlated subsystems, which we denote by $A$ and $E$, are considered and their joint state is given by $\rho_{AE}$. Then, $A$ is affected by a unitary evolution describing Brownian motion. Subsequently a completely positive map $\mathcal{T}$ maps it into another system $B$. The map $\mathcal{T}_{A \rightarrow B}$ can be equivalently described by its  \emph{Choi-Jamiolkowski isomorphism} \cite{Choi}, given by
\begin{equation}
\tau_{A'B}=\left(\mathcal{I}_{A'} \otimes \mathcal{T}_{A \rightarrow B} \right)\left(\ket{\psi}\bra{\psi} \right)_{A'A}
\end{equation}
with
\begin{equation}
  \ket{\psi}_{A'A}=\left|A\right|^{-1/2}\,\sum_j \ket{j}_{A'}\otimes\ket{j}_A
\end{equation}
being a maximally entangled state vector
and $A'$ a copy of $A$ and $|A|$ denoting the Hilbert space dimension of $A$.

A decoupling theorem quantifies the distance, in terms of the 1-norm, of the final state of the above described evolution from the product state $\tau_B\otimes \rho_E$, considering the expectation over the unitary distribution. The bound on this expression is characterised by entropy measures of the initial state $\rho_{AE}$ and of the Choi-Jamiolkowski representation of the map $\mathcal{T}$. More specifically, the entropy measure is the  \emph{conditional collision entropy of $A$ given $B$}, defined as (see ref.~\cite[Definition 2.9]{OneshotDecoupling})
\begin{equation}
H_2 (A|B)_\rho \coloneqq\sup_{\sigma_B \in \mathcal{S}_=(\mathcal{H}_{B})} - \log \tr \left[\left(\left(\1_A \otimes \sigma_B^{-1/4}\right)\rho_{AB}\left(\1_A \otimes \sigma_B^{-1/4}\right)\right)^2\right].
\end{equation}

\begin{figure}
	\begin{center}
	\includegraphics[width=0.6\linewidth]{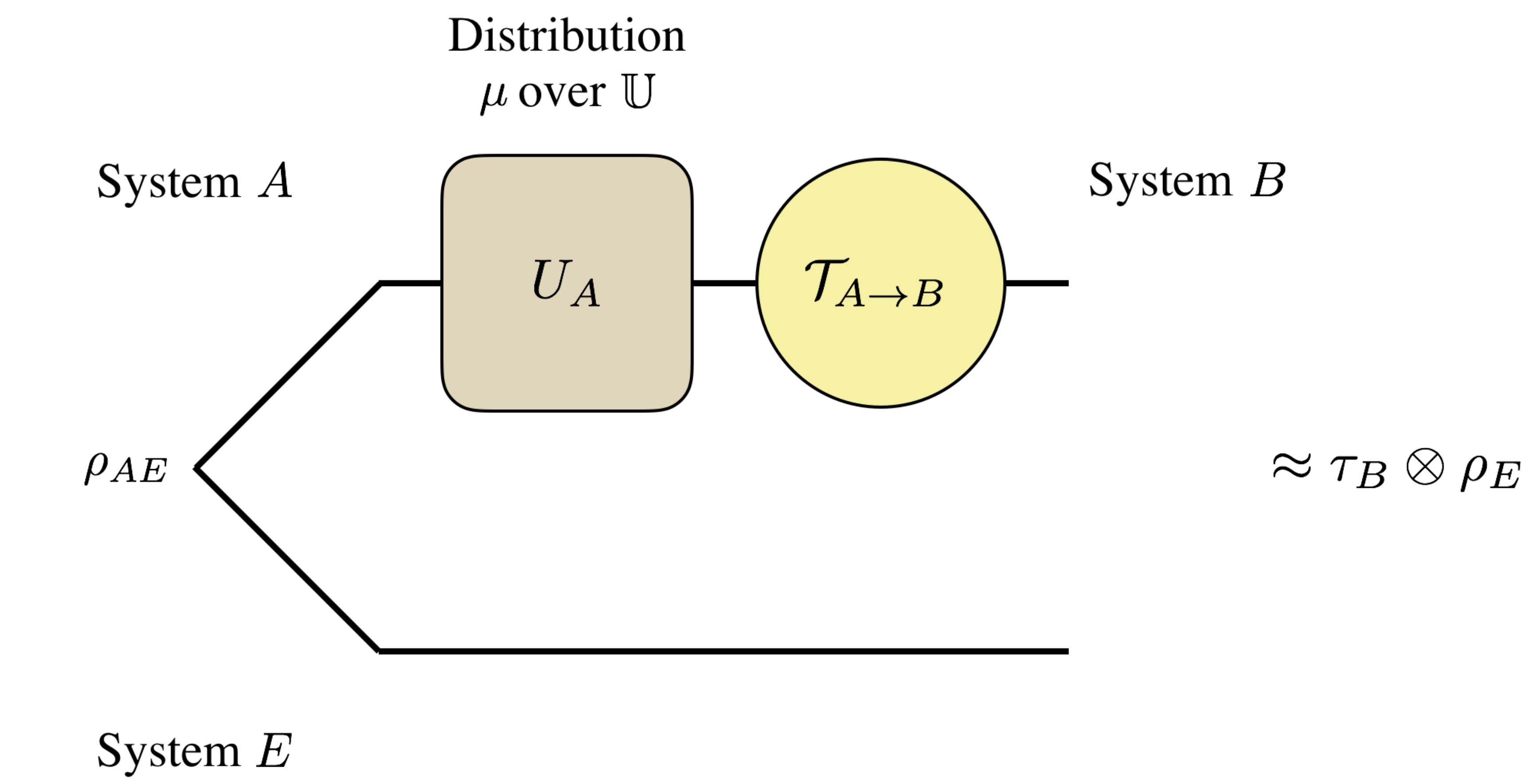}
\end{center}
	\caption{In the decoupling theorem, an initial bipartite state $\rho_{AE}$ is affected by a unitary evolution $U_A$ chosen at according to a certain distribution $\mu$. Then, subsystem $A$ is mapped to another subsystem $B$ through a completely positive map $\mathcal{T}_{A \rightarrow B}$. Finally, the distance between the final state and the product state $\tau_B \otimes \rho_E$ is characterised by entropy measures.}\label{fig:decoupling}
\end{figure}

Our second main result states that, under a unitary evolution describing Brownian motion, decoupling is achieved with a run time scaling almost linear in system size.
We denote the Pauli basis coefficients after a continuous-time evolution with run time $T$ as
\begin{equation}\label{def:Pauli_coefficient}
Q^T(\mu,\nu) \coloneqq \frac{1}{4^n} \Tr\left[\sigma_{\nu} \otimes \sigma_{\nu}\,M_{n,\SLH (T)}^{k=2} (\sigma_\mu \otimes \sigma_\mu)\right]\,.
\end{equation}
First of all, we derive an upper bound on the distance between the distribution of these Pauli coefficients and a distribution which is close to the uniform one.

\begin{theorem}[Mixing condition for Pauli coefficients]\label{thm_4.1}
	For any constants $\delta \in (0,1/16)$, $\eta \in (0,1)$ there exist constants $\varsigma>0$ and $0 <\gamma_0\leq 1/2$ such that for a total run time $T\geq \varsigma \;n\,\log^2 n$ and large enough $n$
	\begin{equation}
	\sum_{\nu\in\left\{0,1,2,3\right\}^n,\nu \neq 0}\left|Q^T(\mu,\nu)-p_\delta(\nu) \right| \leq \frac{1}{(3-\eta)^\ell\binom{n}{\ell}}\frac{1}{{\rm poly}(n)}
	\end{equation}
	where $\sigma_\mu$ is an arbitrary string whose support has size $\ell$ and has a subset of $(1-\gamma)n$ qubits, with $ \gamma<\gamma_0$, which is invariant with respect to any permutation, and $p_\delta$ is a (possibly sub-normalised) distribution on Pauli strings such that:
	\begin{equation}
	p_\delta(\nu) \leq \frac{5^{\delta n}}{4^n-1} \hspace{0.5cm} \forall \nu.
	\end{equation}
\end{theorem}

From Theorem \ref{thm_4.1} we obtain the final result on decoupling, which can be seen as a statement unifying the description of continuous and
	discrete processes. Specifically, it links the stochastic local Hamiltonian evolution (up to the permutation invariance property assumed for the initial state)
	to that of a random quantum circuit under Haar distribution given in ref.~\cite[Theorem 3.2]{Decoupling}, and establishes
	a connection between the discrete random walk induced by a random quantum circuit and the continuous-time random walk implied by the stochastic local Hamiltonian evolution. The actual proof of Theorem \ref{Decoupling_main} from Theorem \ref{thm_4.1} follows by arguments analogous to the random quantum circuit case given in  ref.\ \cite{Decoupling}.
	
\begin{theorem}[Fast decoupling]\label{Decoupling_main}
	Consider a bipartite quantum state $\rho_{AE} \in \mathcal{S}_{AE}$ of an $n$-qubit system $A$ coupled with some other system $E$. Let then $\rho_{AE}$ undergo a unitary evolution $U_t$	induced by stochastic local Hamiltonian increments as in eq.~\eqref{eq:H_Delta_t} acting upon system $A$,  followed by a completely positive trace preserving map $\mathcal{T}:\mathcal{S}_{A} \rightarrow \mathcal{S}_{B}$ which maps from $A$ to another system $B$.
	Let $\tau_{A'B}$ denote the
	Choi-Jamiolkowski  isomorph of $\mathcal{T}$.
	Then, for any $\delta \in(0,1/16)$ there exist $\varsigma>0$ and $0<\gamma_0\leq 1/2$ such that for all $\gamma$-permutation invariant states with $\gamma<\gamma_0$ and total run times $T \geq \varsigma \, n\, \log^2 n$ and for large enough $n$
	\begin{equation}
	\E_{\SLH (T)} \left\{ \left\| \mathcal{T}\left(U_T\,\rho_{AE}U^{\dagger}_T\right) - \tau_B\otimes \rho_E \right\|_1 \right\}
	\leq
	\left({\frac{1}{{\rm poly}(n)} + 5^{\delta n} \cdot 2^{-H_2(A|B)_{\tau}-H_2(A|E)_{\rho}}}\right)^{1/2},
	\end{equation}
	where $\SLH (T)$ denotes the distribution over the unitary group induced by the Brownian motion with run time $T$.
\end{theorem}

\paragraph{Proof idea of Theorem~\ref{thm_4.1}}
	As a first step, we observe the evolution of the support size of Pauli strings during the continuous-time process, looking for the probability to reach a support size within the interval $\left[3/4-\delta,3/4+\delta \right]$ when starting from an arbitrary Pauli string of weight $\ell$.
	The key point is that the jumps of the continuous-time random walk on the Pauli weights induced by the fluctuating Hamiltonian correspond exactly to the accelerated Markov chain (i.e., the chain conditioned on moving) deduced from the random quantum circuit with Haar distribution as in refs.~\cite{Decoupling,RandomCircuitsLow}. This displays an equivalence between the two settings and allows us to link results relating to the two types of processes.
	We then compute the contribution of each string to this probability to obtain the desired condition for Pauli coefficients.
	More precisely, we do so by using the  permutation invariance property and the uniform randomisation of the
	Pauli basis $\sigma_1,\sigma_2,\sigma_3$ -- the latter achieved by the evolution itself. We can then
	infer that almost all Pauli strings with the same support size share the same probability.
	\proofend

\section{Applications}

We discuss in the following two interesting applications for Brownian motion on the unitary group, namely dissipative dynamics and black holes scrambling, and hint at a third one, making use of fluctuating Hamiltonian dynamics in quantum information processing.

\subsection{Dissipative dynamics arising from fluctuations}

As pointed out before, there is an intimate relationship between time-fluctuating dynamics and \emph{Markovian dissipative evolution}, a connection that will be made manifest in this subsection. 
Brownian motion $U_t$ on the unitary group yields an average dynamics given by 
\begin{equation}\label{eq:rho_t}
\rho (t) \coloneqq \EE [ U_t \rho\, U_t^\dagger ] \, ,
\end{equation}
which describes a dissipative quantum Markovian evolution of the state $\rho$.
In this sense, \emph{time-fluctuating classical noise} is precisely a specific source of \emph{dissipation}. As pointed out in the introduction, processes of this kind are ubiquitous in nature and in quantum systems in the laboratory: they originate whenever one does not have perfect control over the
classical control parameters. At the same time, they can be used as a tool. Indeed, the use of \emph{controlled 
dissipative Markovian dynamics} has received much interest in recent years \cite{DissipationZoller,VerWolCir09,KasWolEis13}.

The generator of the dynamical semi-group given by the evolution \eqref{eq:rho_t}, and more generally of the $k$-th moment operator from eq.~\eqref{eq:k-moment}, is calculated below in Lemma~\ref{lem:generator}.
For $k=1$, the generator is explicitly given in the following.

\begin{proposition}[Fluctuations as dissipative processes]
\label{prop:lindblad}
Let $U_t$ be a Brownian motion with increments $\Theta_{\Delta t}$ as in eq.~\eqref{eq:pseudo_derivative}.
Write $\Theta_{\Delta t}$ as
\begin{equation}\label{eq:H0F}
\Theta_{\Delta t}
=
-\i \, H_0 + F_{\Delta t}\, ,
\end{equation}
where $-\i\, H_0$ and $F_{\Delta t}$ are its anti-Hermitian time constant and fluctuating parts, respectively, with
\begin{align*}
F_{\Delta t} = \sum_\mu B_\mu \, \xi^\mu_{\Delta t},
\qquad B_\mu^\dagger = -B_\mu \, ,
\qquad
\EE[\xi^\mu_{\Delta t}] = 0\, ,
\qquad \text{and}\qquad
\EE[\xi^\mu_{\Delta t} \, \xi^\nu_{\Delta t}] = - \frac{a}{{\Delta t}} \, \delta_{\mu,\nu} \, .
\end{align*}
Then $\rho(t) = \EE[U_t \rho U_t^\dagger]$ gives rise to a quantum dynamical semi-group and evolves according to the Lindblad equation
  \begin{equation}
    \frac{d}{dt} \rho(t)
    =
    -\i\, [H_0, \rho(t)] - a \sum_\mu \left(
      B_\mu \rho B_\mu^\dagger
       -\kw 2 \left(B_\mu^\dagger B_\mu \rho+ \rho B_\mu^\dagger B_\mu\right)
      \right)
  \end{equation}
with $\rho(0) = \rho$.
\end{proposition}
This proposition is proven after Lemma~\ref{lem:generator}.

\subsection{Applications for fast scrambling}

\begin{figure}[b]
\begin{center}
\includegraphics[width=0.7\linewidth]{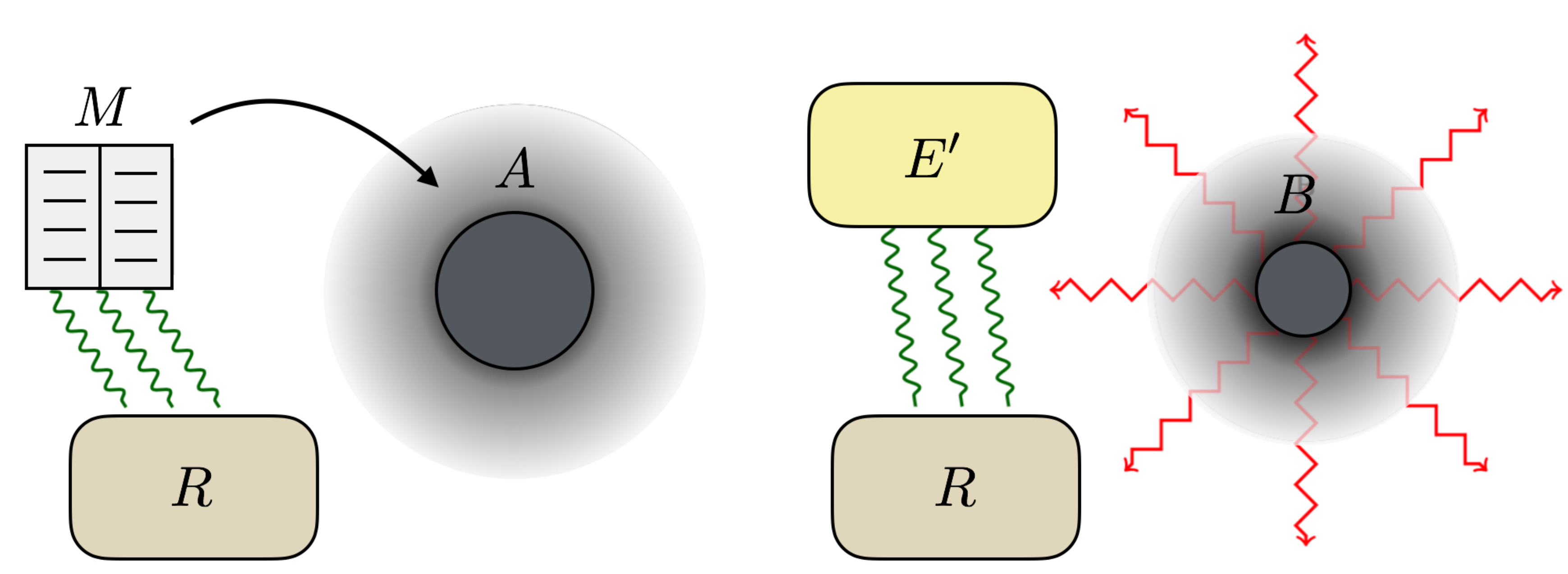}
\end{center}

	\caption{A quantum memory system $M$ is initially entangled with a reference system denoted by $R$ and is subsequently thrown into a black hole $A$ (left picture). As the black holes leaks out Hawking radiation, it shrinks into a smaller system $B$. When a controlled subystem of the irradiated environment $E'$, has become maximally entangled with the reference system $R$, the initial information $M$ has been mirrored (right picture).}
	\label{fig:black_hole}
\end{figure}

In the last decade, black holes have been considered from a quantum information perspective, providing toy models and a fresh
perspective to the field. In particular it has been conjectured that they are \emph{fast scramblers} \cite{BHMirrors,FastScramblingConjecture,LloydPres14}. A system is \emph{scrambled} when any previous perturbation has been thoroughly spread among the degrees of freedom so that to recover information contained in the perturbation one should access simultaneously a large fraction of the entire system. The minimum time for mixing information is then called \emph{scrambling time}. More specifically, in ref.\ \cite{FastScramblingConjecture} three hypotheses have been outlined: the most rapid scramblers take logarithmic time in the degrees of freedom, the bound is saturated for matrix quantum mechanics, i.e., systems whose degrees of freedom are $n \times n$ matrices, black holes are the fastest scramblers in nature. The authors of ref.\ \cite{FastScrambling} brought evidence about the conjectures regarding scrambling in logarithmic time by investigating Brownian quantum circuit and Ising model on sparse random graphs.
There are two related mixing conditions for unitary dynamics that satisfy the requirements for scrambling, as discussed in refs.\ \cite{BHMirrors, FastScramblingConjecture}, or ref.\ \cite{LloydPres14}, respectively. The relation between our results and both of these conditions will be discussed in the following.

In ref.~\cite{BHMirrors}, one considers the black hole's internal system $A$ and the radiated environment $E$. Furthermore, one defines an additional reference system $R$, initially maximally entangled with a quantum memory system $M$ that is subsequently 
thrown into the black hole. As the Hawking radiation leaks out, we would like $R$ to become maximally entangled with a subsystem of $E$ on which we can have control, hence having recovered the initial state of $M$, and so interpreting the black hole as a mirror (see Fig.\ \ref{fig:black_hole}). This may be translated 
into a scrambling condition through a decoupling theorem. As the black hole evaporates, $A$ shrinks into a smaller system $B$ which decouples from $R$. More
formally, this means that 
\begin{equation}
\E_\Haar \left\{ \left\| \Tr_{A \backslash B} \left(U_A\,\rho_{AR}U^{\dagger}_A\right) - \frac{\1_{B}}{|B|} \otimes \rho_R \right\|_1 \right\} 	
 \leq  
 2^{-\gamma},
\end{equation}	
where  $\rho_{AE} $ is a quantum state where subsystem $E$ shares $m$ Bell pair with $A$, and $A$ is otherwise mixed, and $\gamma$ is the difference between the number of qubits emitted as Hawking radiation and the number of qubits of system $M$. The approximate statement 
%(see \cite[Theorem 1]{SzeDuTomRen13})
\begin{equation}
\E_\omega \left\{ \left\| \Tr_{A \backslash B} \left(U_A\,\rho_{AR}U^{\dagger}_A\right) - \frac{\1_{B}}{|B|} \otimes \rho_R \right\|_1 \right\} 	
\leq  
\sqrt{4^{-\gamma} + 4^{m}\varepsilon}
\end{equation}		
 is satisfied in expectation for an ensemble of unitary transformations $\omega$, which is an approximate 2-design in the sense that the Pauli coefficients are close to the uniform distribution, i.e.
\begin{equation} 
\sum_{\nu \neq 0} |q_\omega(\mu,\nu) - q_u(\mu, \nu)| \leq \varepsilon  \qquad \forall \mu \ ,
\end{equation}
where 
\begin{align}
q_u(\mu, \nu)= \frac{1}{4^n-1} \quad \forall \mu,\nu  
&& \text{and} && 
q_\omega(\mu,\nu) 
=
\frac{1}{4^{n}} \Tr\left[\sigma_{\nu} \otimes \sigma_{\nu}\,M_{\omega}^{k=2} (\sigma_\mu \otimes \sigma_\mu)\right]\,.
\end{align} 		
This condition was shown in ref.\ \cite{RandomCircuitsLow} to be satisfied by a random quantum circuit of size $O(n\log n)$ (when $\epsilon=1/\mathrm{poly}(n)$) and analogously by a stochastic local Hamiltonian, according to the analysis on random walk in Section \ref{sec:Decoupling_poof} and following the same reasoning as in ref.\ \cite{RandomCircuitsLow}, with a run time $T=O(n\log n)$.
However, in order to compare time scales with  ref.\ \cite{FastScrambling}, we take the same convention and divide the global scrambling time by the time to scramble a single subsystem ; in this case we obtain a scrambling time of $\tau_\ast = O(\log n).$ Hence, our work also provides an alternative proof for the scaling of the scrambling time in ref.\ \cite{FastScrambling}, although our argument does not involve any intermediate conjectures, such as the final statements of 
ref.\ \cite[Appendix B]{FastScrambling}.

In ref.~\cite{LloydPres14}, a slightly different scrambling condition is required for the unitarity of black hole evaporation to hold, given postselection on the final state at the singularity inside the black hole.  One considers the composite system $\mathcal{H}_M \otimes \mathcal{H}_{\mathrm {in}} \otimes \mathcal{H}_{\mathrm {out}}$ representing the infalling matter, the infalling negative energy Hawking radiation behind the event horizon and the outgoing positive energy Hawking radiation outside the horizon, respectively. Again, one defines a reference system $S$ which is maximally entangled with a subsystem $M_1 \subset M$. After the application of a random unitary transformation $U$ on $\mathcal{H}_M \otimes \mathcal{H}_{\mathrm {in}}$ and subsequently tracing out the complement subsystem of $S$, we have (cf. \cite[eq.(3)]{LloydPres14}):
\begin{equation}
\E_\Haar \left\{ \left\| \Tr_{\overline{S}} \left(U\, \rho \, U^\dagger \right) - \frac{\1_S}{|S|} \right\|_1 \right\}	
\leq
\sqrt{\frac{|\mathcal{H}_{M_1}|}{|\mathcal{H}_{in}|}}.
\end{equation}
A relaxed version of this bound, namely
\begin{equation}
\E_\omega \left\{ \left\| \Tr_{\overline{S}} \left(U\, \rho \, U^\dagger \right) - \frac{\1_S}{|S|} \right\|_1 \right\}	
\leq
\sqrt{5^{\delta n}\frac{|\mathcal{H}_{M_1}|}{|\mathcal{H}_{in}|} + \frac{1}{{\rm poly}(n)}} \ ,
\end{equation}
where $n = \log_2 \left( |\mathcal{H}_{in}| |\mathcal{H}_{M}| \right)$, follows from the condition
\begin{equation}
\sum_{\nu \neq 0} |q_\omega(\mu,\nu) - 4^{\delta n} q_u(\mu, \nu)| 
\leq 
\frac{1}{(3-\eta)^\ell\binom{n}{\ell}}\frac{1}{{\rm poly}(n)}  \ ,
\end{equation}
for every Pauli string $\sigma_\mu$ with support size $\ell$ .\\

The above condition was shown to hold in ref.\ \cite{Decoupling} for random quantum circuits of size $O(n \log^2 n).$   Applying the equivalence established Section\ \ref{sec:Decoupling_poof}, it follows from Theorem~\ref{thm_4.1} that this is fulfilled by a stochastic local Hamiltonians in time $\tau'_\ast= O(\log^2 n)$,  when again we take the convention of ref.\ \cite{FastScrambling} and divide global scrambling time by the time to scramble a single subsystem.  

Note finally that recently, an interesting connection between
 chaos -- as being captured by \emph{out-of-time-order correlation functions} -- and pseudorandomness -- as formalised in the 
 notion of a unitary design -- has been established \cite{Chaos}. Invoking those results, the findings presented here on
 generating approximate unitary designs by making use of time-fluctuating dynamics can be applied to assess quantum chaos in this sense.

\subsection{Applications in quantum information processing}

We finally mentioned a third, immediate, application, which seems yet particularly important when having potential technological
applications in quantum information processing
in mind. It should be clear that whenever the aim is to realize an approximate unitary design, the evolution under a 
fluctuating Hamiltonian already constitutes a valuable option. In many domains of quantum information, specifically in
notions of benchmarking, approximate unitary designs are important primitives
\cite{tomographydesigns,Efficient,Benchmarking}.
It is known that with a suitable random 
circuit one can generate an $\varepsilon$-approximate $k$-design
\cite{BrandaoHarrowHorodecki}. Such a prescription,
however, requires the precise implementation of a deep quantum circuit consisting of a large number of local quantum gates, namely  $O(n k^9 (n k  + \log 1/\varepsilon))$.
The above results have the interesting implication: instead of implementing a quantum circuit, a suitably
stochastic Hamiltonian evolution gives rise to exactly the same dynamics. In such an approach, the classically fluctuating
parameters would have to be stored. This insight might have important applications in quantum information processing.

%Note that in Theorem~\ref{thm_TPE} we have chosen the local increments \eqref{eq:stochastic_local_H_details} to be instead independent of $n$. In this case, as we will see, the effective local generator \eqref{eq:def_g_k} turns out to depend on $n$.

%===========================================
%===========================================
%===========================================
%===========================================
\section{Stochastic time evolution generates tensor product expanders} \label{section:designs}
%===========================================
%===========================================
%===========================================
%===========================================
In this section, we prove Theorem~\ref{thm_TPE} bounding the time after which the stochastic time evolution becomes a tensor product expander.
As a crucial step we investigate the gap of the local generator induced by the Hamiltonian increments as given in eq.~\eqref{eq:stochastic_local_H}.

The proof will be structured as follows: we first derive in Lemma~\ref{lem:generator} the generator of the $k$-th moment operator and then describe how this allow us to express it as a tensor product expander using previous results on random quantum circuits.
In Subsection \ref{sec:LocGap} we provide the central mathematical results of this work, namely a diagonalisation of the local generator by relating it to the Casimir element in the enveloping algebra of $\su(d^2)$. Since only certain irreps are contained in the direct sum decomposition of the Casimir element, we will observe that no eigenvalue  can assume a value in the interval (0,1), giving rise to a local gap.

Much of the developed machinery will build upon the
representation theory of the special unitary group. 
%A \emph{representation} of a group $G$ on a vector space $V$ is a homomorphism
%\begin{equation}
%\pi
%\colon
%G \ \rightarrow \ GL(V),
%\end{equation}
%where $GL(V)$ denotes the general linear group over $V$, and is said to be  \emph{faithful} if it is injective. A subspace $W \subset V$ is said to be  \emph{invariant} if for all $g \in G$ and $w \in W$
%\begin{equation}
%\pi(g)w \in W.
%\end{equation}
%A representation is said to be  \emph{irreducible} if the only invariant subspaces are $\{0\}$ and $V$ itself. Every complex representation of a finite group is  \emph{completely reducible}, i.e. it can be decomposed as a direct sum of irreducible representations.
It will also be helpful to use the identification of maps on matrices with matrices
(induced vectorisation of matrices) given by
$\vect(XYZ) = (X\otimes Z^T) \vect(Y)$
to express the $k$-th moment operator as
\begin{equation}
\label{Mrep}M_{\mu}^k = \EE_\mu [\pi_{k,k}(U)] \, ,
\end{equation}
where $\pi_{k,k}(U)$ is the $(k,k)$-mixed tensor representation of the group element $U \in \mathds{SU}(N)$ given by
\begin{equation} \label{eq:pi_kk}
\pi_{k,k}(U)\coloneqq  U^{\otimes k}\otimes\overline{U}^{\otimes k}.
\end{equation}
We also make use of the corresponding representation of the Lie algebra
$\su(N)$ which is also denoted by $\pi_{k,k}$ and satisfies the following for all $X \in \su(N)$
\begin{equation}
\pi_{k,k}(\exp(X)) = \exp(\pi_{k,k}(X)),
\end{equation}
with
\begin{equation}\label{eq:pi_kk_def}
\pi_{k,k}(X) = \sum_{i=1}^k X\otimes \1_{\overline{i}}  + \sum_{i=k+1}^{2k} \overline X \otimes \1_{\overline{ i}}
\end{equation}
and
\begin{equation}
X\otimes \1_{\overline{i}} \coloneqq \1_1\otimes \ldots \otimes \1_{i-1} \otimes  X \otimes \1_{i+1} \otimes \ldots \otimes \1_{2k}.
\end{equation}
This representation plays a central role in our analysis of the gap of the $k$-th moment operator of the stochastic time evolution.

The $k$-th moment operator $M^k_{\SLH(T)}$ has a generator that we explicitly calculate in the following. In fact, the lemma also holds for general Brownian motions on $\U(N)$, not only the locally generated ones considered in our theorems.

\begin{lemma}[The generator of the $k$-th moment operator]\label{lem:generator}
Let $M_T^k$ be the $k$-th moment operator of a unitary Brownian motion with increments $\Theta_{\Delta t}$ as in eq.~\eqref{eq:pseudo_derivative} at time $T$. Then
\begin{equation}
  M_T^k = \e^{G^k\, T}
\end{equation}
with
\begin{equation}\label{eq:loc_g}
  G^k = \lim_{ t \to 0} \left(\EE\left[\pi_{k,k}(\Theta_t)\right]
	    +\frac 1 2 \EE\left[
	    \pi_{k,k}(\Theta_t)^2\,  t
    \right] \right) .
\end{equation}
\end{lemma}
Note that as $\Theta_t$ is anti-Hermitian, $G^k$ is negative semidefinite. If the Brownian motion is universal then the kernel of $G^k$ is the invariant subspace of $M^k$.

Most steps in the proof of this lemma will be also used again in the proof of Theorem~\ref{thm_TPE}.

\begin{proof}
	As $M_T^k$ is a Markov process,
	% for $T= s\, \Delta t$ with $s \in \mathds{Z}^+$,
	we have
	\begin{equation}\label{eq:composing_Dt}
	M_T^k = (M_{\Delta t}^k)^{T/\Delta t}  \, .
	\end{equation}
	With the mixed tensor representation \eqref{eq:pi_kk} and the definition of the $k$-th moment operator \eqref{eq:k-moment} we obtain for a single time step
	\begin{align}
	M_{\Delta t}^k
	&=
	\EE\left[ \pi_{k,k}\left( U_{\Delta t}\right) \right]
	\\
	&=
	\EE\left[ \pi_{k,k}\left(
	\exp\left\{
	%        \sum_{e \in E}\sum_\mu A_\mu^{(e)} \xi^{(e,\mu)}_{\ell,\Delta t}
	\Theta_{\Delta t}
	\, \Delta t \right\} \right) \right]
	\end{align}
	with $U_{\Delta t}=\e^{\Theta_{\Delta t}\, \Delta t}$ and the increments $\Theta_{\Delta t}$ from eq.~\eqref{eq:pseudo_derivative}.
	Using a Taylor expansion yields
	\begin{align}
	M_{\Delta t}^k
	& =
	% 	\EE\left[\left(\e^{ \Theta_{\Delta t}\Delta t}\right)^{\otimes k,k} \right]
	% 	=
	\EE\left[\e^{\pi_{k,k}(\Theta_{\Delta t}) \Delta t}\right] \\
	&=
	\sum_{p=0}^\infty\frac{(\Delta t)^p}{p!}\EE\left[\pi_{k,k}(\Theta_{\Delta t})^p \right] \\
	&=
	\1_{\dim(\pi_{k,k})}
	+ \EE\left[\pi_{k,k}(\Theta_{\Delta t})\right]\Delta t
	+\EE\left[\pi_{k,k}(\Theta_{\Delta t})^2\right]\frac{\Delta t^2}{2}+\O(\Delta t^2)
	\, . \label{eq:M_taylor}
	\end{align}
	Composing the time steps as in eq.~\eqref{eq:composing_Dt}, we obtain
	\begin{align}
	M_T^k
	=
	\lim_{\Delta t \to 0} \left( \1_{\dim(\pi_{k,k})}
	+ \left(\EE\left[\pi_{k,k}(\Theta_{\Delta t})\right]
	+ \EE\left[\kw 2 \pi_{k,k}(\Theta_{\Delta t})^2 \Delta t \right]\right) \Delta t \right)^{T/\Delta t}
	\end{align}
	and  finishes the proof.
	\proofend
\end{proof}

As a direct application of this lemma and as an exercise for the proof of Theorem~\ref{thm_TPE}, we prove Proposition \ref{prop:lindblad} in the following.

\begin{proof}[Proposition \ref{prop:lindblad}]
	According to Lemma~\ref{lem:generator} the evolution has a generator
	\begin{align}
	G^1&
	=
	\lim_{t\to 0} \EE\left[\Theta_t\otimes \1 + \1 \otimes \overline \Theta_t\right]
	+ \kw 2 \lim_{t\to 0}
	\EE[(\Theta_t\otimes \1 + \1 \otimes \overline \Theta_t)^2] \, t \, ,
	\end{align}
	where the mixed tensor representation $\pi_{1,1}$ from eq.~\eqref{eq:pi_kk_def} is used.
	It remains to show that the generator is of Lindblad form.
	The proposition's hypothesis yields
	\begin{align}
	G^1
	&=
	-i\, H_0\otimes \1 + \i\, \1 \otimes \overline H_0
	+ \kw 2 \lim_{t\to 0}
	\EE[(F_t^2\otimes \1 + \1 \otimes \overline F_t^2
	+2 F_t \otimes \overline F_t)]\, t
	\\
	&=
	-i\, H_0\otimes \1 + \i\, \1 \otimes H_0^T
	- \frac{a}{2} \sum_\mu
	(B^2_\mu\otimes \1 + \1 \otimes B^{2\, T}_\mu
	-2 B_\mu \otimes B^{T}_\mu) \, ,
	\end{align}
	where we have used that $\overline B_\mu = - B_\mu^T$.
	The identification $\vect(XYZ) = (X\otimes Z^T) \vect(Y)$ and
	$B_\mu^\dagger= -B_\mu$
	finish the proof.
	\proofend
\end{proof}

We note the following theorem on random quantum circuits generated by general local distributions, which is implicitly contained in ref.\ \cite[Corollary 7]{BrandaoHarrowHorodecki}.
\begin{lemma}[Relating global and local gaps]\label{lem:localgabvsglobalgap}
  Let $\mu_{loc}$ be a distribution on $\U(d^2)$ and $\mu$ be the distribution on $\U(d^n)$ that applies a unitary drawn according to $\mu_{loc}$ to a uniformly chosen edge of an interaction graph $i, i+1$. Then, its moment operator satisfies
  \begin{align}
    \norm{M_\mu^k - M_\Haar^k}_\infty \leq 1 - \left(1-\norm{m^k_{\mu_{loc}} - m^k_{\Haar_{loc}}}_\infty\right)\left(1-\norm{M^k_{\text{RQC},\Haar}-M_\Haar^k}_\infty\right)\, ,
  \end{align}
  where $\Haar_{loc}$ denotes the Haar measure on $\U(d^2)$ and $M^k_{\text{RQC},\Haar}$ the moment operator of the unitary circuit with distribution $\mu = \Haar_{loc}$.
\end{lemma}
\begin{proof}
  Setting $\left(M_\Haar^k\right)^\perp  = \1-M_\Haar^k$ we have the relation
  \begin{align}\label{eq:genlocmeasure}
     \norm{M_\mu^k - M_\Haar^k}_\infty &=  \norm{\left(M_\Haar^k\right)^\perp M_\mu^k\left(M_\Haar^k\right)^\perp}_\infty
     \\&= \norm{\frac{1}{\abs{E}}\sum_{e\in E}^{n-1}  \left(M_\Haar^k\right)^\perp m_{\mu_{loc}}^{k,(e)}\left(M_\Haar^k\right)^\perp}_\infty\; .
  \end{align}
  By denoting $\gamma \coloneqq \norm{m_{\mu_{loc}}^k-m_{\Haar_{loc}}^k}_\infty$, we find $m_{\mu_{loc}}^{k,(e)} \leq (1-\gamma)\, m_{\Haar_{loc}}^{k,(e)} + \gamma \1$, which implies the  operator inequality
  \begin{align}
  \left(M_\Haar^k\right)^\perp\left(\sum_{e\in E} m_{\mu_{loc}}^{k,(e)}\right)\left(M_\Haar^k\right)^\perp\leq
   \;&\gamma \abs{E} \left(M_\Haar^k\right)^\perp\\  &+ (1-\gamma) \left(M_\Haar^k\right)^\perp\left(\sum_{e\in E}m_{\mu_{\Haar_{loc}}}^{k,(e)}\right)\left(M_\Haar^k\right)^\perp\; .
  \end{align}
  Since $(1-\gamma)$ is positive, we can use the bound $A\leq\norm{A}_\infty$ for the second summand on the right hand side which, together with 
  \begin{equation}
  \norm{M^k_{\text{RQC},\Haar}- M^k_\Haar}_\infty = \norm{\left(M_\Haar^k\right)^\perp 
  M^k_{\text{RQC},\Haar} \left(M_\Haar^k\right)^\perp}_\infty
  \end{equation} 
 finishes the proof. \proofend
\end{proof}

Now we present the main proof of Theorem~\ref{thm_TPE}.
Part of it will be completed with the lemmas stated and proved subsequently.
\begin{proof}[Proof of Theorem~\ref{thm_TPE}]
Thanks to Lemma~\ref{lemma:k-copy_gapped} it is enough to bound the gap of the $k$-th moment operator $M^k_{\SLH(T)}$.
According to Lemma~\ref{lem:driving}, the time constant part of the Hamiltonian does not affect the invariant subspace nor the gap of  $M^k_{\SLH(T)}$.
Hence, we can set without loss of generality $h^{(e)}_0=0 \, \forall e$.
Additionally, in Lemma~\ref{lem:CompleteGraph} we prove that the gap of an interaction graph being a complete graph is larger than the one of a 1D graph.
We hence consider only the latter case in the proof.

Using the approximation \eqref{eq:M_taylor} and expressing $\Theta$ in terms of the local terms $\theta^{(e)}$ (as in eq.~\eqref{eq:stochastic_local_H_details}) we obtain
% \begin{align}\label{generator_global_moment}
% 	M^k_{\SLH (T)} &= \lim_{\Delta t \to 0} (M_{\Delta t}^k)^{T/ \Delta t}
% 	\\
% 	& = \lim_{\Delta t\rightarrow 0} \left( \1_{\dim(\pi_{k,k})} + \EE[\pi_{k,k}(\Theta_{\Delta t})^2] \frac{\Delta t^2}{2}\right )^{T/\Delta t}
% 	\\
% 	&= \lim_{\Delta t \rightarrow 0} \left(\1_{\dim(\pi_{k,k})} +
% 		\sum_{e\in E}  \EE\left[\pi_{k,k}(\theta^{(e)} (\Delta t))^2\right] \frac{\Delta t^2}{2}\right)^{T/\Delta t}
% \end{align}
\begin{align}\label{generator_global_moment}
	M^k_{\SLH (\Delta t)}
	&= \1_{\dim(\pi_{k,k})} +
		\sum_{e\in E}  \EE\left[\pi_{k,k}(\theta^{(e)}_{\Delta t})^2\right] \frac{\Delta t^2}{2} + O(\Delta t^2)
\end{align}
Using another Taylor approximation yields
\begin{align}
M^k_{\SLH (\Delta t)} 	
&=
	 \1_{\dim(\pi_{k,k})} +\frac{1}{n} \sum_{e\in E}  \EE\left[\pi_{k,k}\bigl(\sqrt{n}\,\theta^{(e)}_{\Delta t} \bigr)^2\right] \frac{\Delta t^2}{2} +O(\Delta t^2)
	 \label{eq:derive_g}
	 \\
	 & =
	 \frac{1}{n} \sum_{e\in E}
	 \E \left[ \left(\exp\{\sqrt{n}\, \theta^{(e)}_{\Delta t} \, \Delta t\} \right)^{\otimes k,k} \right]
	 +O(\Delta t^2)  \,  .
	 \label{eq:derive_g_2}
% 	 \\
% 	  & =
% 	 \frac{1}{n} \sum_e \E \left[\left( \left(\exp\{\, \theta^{(e)}  \Delta t\} \right)^{\otimes k,k} \right)^{\mk{\sqrt{n}}}\right]+O(\Delta t^2)
\end{align}
% where the form of the covariance of the noise \eqref{eq:noise_cov} has been used to obtain the rescaling of $\Delta t$ with~$n$.

Next, we view
$G(\theta^{(e)}_{\Delta t})
\coloneqq
\exp\{\sqrt{n} \,\theta^{(e)}_{\Delta t}\, \Delta t\}$
as a random gate in a
\emph{$G$-random quantum circuit} considered in ref.~\cite{BrandaoHarrowHorodecki}.
The (system size independent) local $k$-th moment operator on edge $e\in E$ is
\begin{align}\label{eq:def_m}
m_{\Delta t}^{k,(e)}
&\coloneqq
  \E \left[ \left(\exp\{\sqrt{n}\,\theta^{(e)}_{\Delta t} \, \Delta t\} \right)^{\otimes k,k} \right]  .
\end{align}
Note that this $k$-th moment operator also corresponds to a Brownian motion but with a variance rescaled by a factor of $n$, cf. also the parameter $a$ in eq.~\eqref{eq:noise_cov}.
As its gap, i.e., the difference between the largest and second largest eigenvalue does not depend on $e$ we simply denote the gap of $m_{\Delta t}^{k,(e)}$ by $\Delta(m_{\Delta t}^{k})$.
Then the local gap lemma \ref{lem:localgabvsglobalgap} yields directly
\begin{align}\label{gap_relation}
	\left\| M^k_{\SLH (\Delta t)}-M^k_{\Haar} \right\|_\infty
	&\leq  	
	1-\Delta\left(m_{\Delta t}^{k}\right) \left(1-\norm{M^k_{\text{RQC},\Haar}-M_\Haar^k}_\infty\right) \, ,
\end{align}
	%where $\Delta(H_{k,\Haar})$ is given by the relation (cf. ref.\ \cite[Lemma 16]{BrandaoHarrowHorodecki})
%	\begin{equation}
%	\left\| M^k_{\text{RQC},\Haar}-M^k_{\Haar} \right\|_\infty
%	=
%	1-\frac{ \Delta(H_{k,\Haar})}{n}
%	\end{equation}
	with $M^k_{\text{RQC},\Haar}$ being the $k$-th moment operator of single step of a local random quantum circuit whose gates are chosen from the Haar measure.
	The gap of $M^k_{\text{RQC},\Haar}$ can be lower bounded as \cite[eq.~(41)]{BrandaoHarrowHorodecki}
	\begin{align}
	\left(1-\norm{M^k_{\text{RQC},\Haar}-M_\Haar^k}_\infty\right)
	&\geq
	\frac{k^{-\frac{2.5}{\ln(d)}-2.5\frac{\ln(d^2+1)}{\ln(d)}}}{125\lceil\log_d(4k)\rceil^2\e (d^2+1)}
	\\
	&\geq
      \frac{1}{425\, n \lceil \log_d(4t)\rceil^2\, d^2\, k^5\, k^{3.1/\ln(d)}}
	\, ,
	\end{align}
	where we have also used the inequalities right after \cite[eq.~(41)]{BrandaoHarrowHorodecki}.
	Together with eq.~\eqref{gap_relation}, these results imply
	\begin{equation}\label{gap_relation_2}
	\norm{M^k_{\SLH (\Delta t)}-M^k_{\Haar}}_\infty
	\leq
	1- \Delta(m_{\Delta t}^{k})
	\, s/n
	\end{equation}
with
\begin{equation}
  s\coloneqq \left(425\lceil\log_d(4k)\rceil^2 d^2 k^5 k^{3.1/ \ln(d)} \right)^{-1}.
\end{equation}
In order to calculate $\Delta(m_{\Delta t}^{k})$ we use Lemma~\ref{lem:generator}, eq.~\eqref{eq:def_m} and 
\begin{equation}
\EE\left[\pi_{k,k}(\theta^{(e)}_t)\right]=0
\end{equation}
so that we can express $m_{\Delta t}^{k,(e)}$ as
\begin{align}
  m_{\Delta t}^{k,(e)}
  &=
  \exp\left(n\, g_k^{(e)}\, \Delta t  \right)
  \\
  &= \1 + n\, g_k^{(e)}\, \Delta t   + O(\Delta t^2)
\end{align}
with
\begin{equation}\label{eq:def_g_k}
\begin{aligned}
  n\, g_k^{(e)} &= \kw 2 \lim_{t \to 0}
	    \EE \left[\pi_{k,k}(\sqrt{n}\, \theta^{(e)}_t)^2\, t \right]
	\\
	&=\frac n 2 \lim_{t \to 0}
	    \EE \left[\pi_{k,k}( \theta^{(e)}_t)^2\, t \right] \ .
\end{aligned}
\end{equation} 
Hence,
\begin{equation}
  \Delta(m_{\Delta t}^{k})
  =
  n \, \Delta(g_k)\, \Delta t
  +O(\Delta t^2) \, ,
\end{equation}
where $\Delta(g_k)$ denotes again the spectral gap to the invariant subspace, i.e., minus the largest non-zero eigenvalue of $g_k^{(e)}$.

As both $k$-th moment operators have the same unit eigenvalue eingenspace according to Lemma~\ref{lemma:k-copy_gapped}, $\norm{M^k_{\SLH (\Delta t)}-M^k_{\Haar}}_\infty$ is the second largest eigenvalue of $M^k_{\SLH (\Delta t)}$. Hence,
\begin{align}
  \norm{M^k_{\SLH (T)}-M^k_{\Haar}}_\infty
  &=
  \lim_{\Delta t \to 0}
  \norm{\bigl(M^k_{\SLH (\Delta t)}\bigr)^{T/\Delta t}-M^k_{\Haar}}_\infty
  \\
    &=
  \lim_{\Delta t \to 0}
  \left(
  1-
    s\, \Delta(g_k)
  \right)^{T/\Delta t}
  \\
  &=
  \exp\left(-T\, s \, \Delta(g_k) \right) \, .
\end{align}
 Observation~\ref{lem:casimir} and Lemma~\ref{lem:local_gap} yields that the gap is the same as the variance \eqref{eq:noise_cov} of the noise, $\Delta(g_k)=a/2$, which completes the proof.
 \proofend
\end{proof}

\subsection{Local gap}\label{sec:LocGap}
In order to calculate the local gap $\Delta(g_k)$, the following representations for the algebra $\su(N)$ will be used.
\begin{alignat}{4}
\text{Trivial rep.}&\quad&& \pi_1 : \su(N) \to \mathfrak{gl}(1,\CC), &\quad& X \mapsto 0,
\\
\text{Fundamental rep.}&&& \pi_f : \su(N) \to \mathfrak{gl}(N,\CC) , && X \mapsto X,
\\
\text{Adjoint rep.}&&& \pi_{\mathrm{ad}} : \su(N) \to \mathfrak{gl}\left(\su(N)\right) , && X \mapsto \ad_X,
\end{alignat}
where  $\ad_X$ is defined by $\ad_X(Y)\coloneqq [X,Y]$.

\begin{observation}[Omitting the phase]
From the mixed-tensor representation we note that we can restrict the analysis on the $\su(N)$ algebra instead of $\u(N)$: the phase factor in the semi-direct product decomposition of any $U \in \U(N) \simeq \SU(N) \rtimes \U(1)$ is cancelled by its complex conjugate coming from $\overline{U}$. In this representation the two algebras are indistinguishable.

\end{observation}

The  \emph{Killing form} $K$ in $\su(N)$ is the symmetric bilinear form defined by
\begin{align}
K(X,Y) &\coloneqq \Tr\left[\ad_X \ad_Y\right] \, .
\end{align}
Denoting the Hilbert-Schmidt inner product of $X$ and $Y$ (in the fundamental representation) by $\langle X,Y\rangle = \Tr(X^\dagger Y)$, the Killing form of $\su(N)$ can also be written as
	\begin{equation}
	    \label{eq:Killing_form_fundamental_rep}
	K(X,Y) = -2N\, \langle X,Y \rangle \, .
	\end{equation}
	
In terms of a basis $\{X_\mu\}_{\mu = 1}^{N^2-1}$ of $\su(N)$ the  \emph{Killing metric tensor} $\kappa$ is defined by
\begin{equation}\label{eq:def_kappa}
  \kappa_{\mu,\nu} \coloneqq K(X_\mu,X_\nu)\, ,
\end{equation}
as was already indicated in eq.~\eqref{eq:Killing_tensor_defining_rep}.
Then, the  \emph{Casimir element} in a matrix representation $\pi$ is
\begin{equation}\label{eq:Casimir}
C(\pi) \coloneqq \sum_{\mu, \nu} \kappa^{-1}_{\mu,\nu}\, \pi(X_\mu)\pi(X_\nu) \, .
\end{equation}

According to eqs.~\eqref{eq:def_g_k} and~\eqref{eq:stochastic_local_H_details}, the local generator $g_k^{(e)}$ of our unitary process with vanishing driving $h_0^{(e)}=0$ is given by
\begin{equation}
g_k^{(e)}
=
\kw 2 \lim_{\Delta t \to 0} \EE\left[\pi_{k,k}\left(\sum_{\mu=1}^{N^2-1} A_\mu^{(e)}\, \xi^{(e,\mu)}\right)^2 \right] \Delta t
=
-\frac{a}{2}\, \sum_{\mu,\nu=1}^{N^2-1} \kappa_{\mu,\nu}^{-1}\mkern2mu \pi_{k,k}\left(A_\mu^{(e)}\right)\pi_{k,k}\left(A_\nu^{(e)}\right) \label{eq:g(e)}
\end{equation}
(where $N\coloneqq d^2$).
The second equality follows from our central assumption~\eqref{eq:noise_cov}.
All $g_k^{(e)}$ are tensor copies of a local operator $g_k$. Therefore, we will suppress the subscripts $e$ in this section from now on.

\begin{observation}[Casimir element]
    \label{lem:casimir}
    Let $g_k$ be the generator of the local $k$-th moment operator in eq.~\eqref{eq:g(e)}. %with $\{A_\mu\}$ being a Hilbert-Schmidt orthonormal basis of $\su(d^2)$.
    Then
    \begin{equation}
	    \label{eq:local_moment_operator}
	    g_k
	    =-
	     \frac{a}{2}\,
	     C(\pi_{k,k}) \, .
	\end{equation}
\end{observation}
% \begin{proof}
%     Immediate from~\eqref{eq:Casimir}.\qed
% \end{proof}

More generally, an overcomplete set $\set{A_\mu}$ can also be admitted, as already mentioned in Remark~\ref{rmk:basis_of_algebra}. The final result about the convergence rate -- up to a constant $O(1)$ -- is still valid as long as the generator and the Casimir element are related by an equation of the form
	\begin{equation}\label{eq:var_generator}
	g_k =-\frac{a}{2}'  C(\pi_{k,k}) + g' \ ,
	\end{equation}
where $a'>0$ and $g'$ is negative semidefinite so that it can only increase the gap.

In the following, we prove that the eigenvalues of the Casimir do not assume a value within the interval $\left(0,1\right)$, for all $k$.

\begin{lemma}[Casimir gap]
    \label{lem:local_gap}
    Let $\mathcal{I}_k$ be the set of  \emph{irreducible} representations occurring in $\pi_{k,k}$ and let $m_k(\pi)\in\mathbb{N}$ denote the multiplicity of each such representation~$\pi$.
    Then
	\begin{equation}\label{eq:Casimir_decomposition}
	C(\pi_{k,k})\simeq \bigoplus_{\pi \in \mathcal{I}_k} c_2(\pi)\mkern2mu \1_{\dim(\pi_{k,k})}\otimes \1_{m_k(\pi)},
	\end{equation}
	where
	\begin{equation}
	c_2(\pi)
	\begin{cases}
	=0 & \text{if $\pi\simeq\pi_1$,} \\
	=1 & \text{if $\pi\simeq\pi_\mathrm{ad}$,} \\
	>1 & \text{otherwise.}
	\end{cases}
	\end{equation}
	In particular, the spectral gap of $C(\pi_{k,k})$ is independent of~$k$.
\end{lemma}
\begin{proof}
	Since the Casimir element is an element of the center of the universal enveloping algebra, from Schur's Lemma follows that it acts as a multiple of the identity in each irreducible representation (see ref.~\cite[Chapter 12]{SattWea}), so
	that~\eqref{eq:Casimir_decomposition} is immediate.
	Now, since the tensor product between the fundamental representation and its conjugate are isomorphic to the direct sum of the trivial and the adjoint ones, this means that the representation~$\pi_{k,k}$ is isomorphic to $(\pi_1\oplus\pi_\mathrm{ad})^{\otimes k}$.
	
	The trivial representation is guaranteed to occur in the decomposition of~$\pi_{k,k}$ into irreducible representations
	(for example, via $\pi_1^{\otimes k}$) and leads to the eigenvalue $c_2(\pi_1)=0$.
	The adjoint representation always occurs -- for example, via $\pi_\mathrm{ad}\otimes\pi_1^{\otimes(k-1)}$ and permutations thereof -- too, and leads to the eigenvalue $c_2(\pi_\mathrm{ad})=1$.
	If we can show that no other irreducible representation~$\pi$ with $c_2(\pi)\leq 1$ occurs, the proof is complete.
	
	One might think that this requires rather detailed knowledge about how tensor product representations of the form $\pi_\mathrm{ad}^{\otimes l}$ decompose into irreducible representations. To follow the next argument, some basic knowledge regarding Young diagrams is necessary; please refer to Appendix \ref{app:Young_diagrams}. 
	It is in fact sufficient to exploit a remarkably basic property which is shared by all the irreducible $\su(d^2)$ representations occurring in~$\pi_{k,k}$: their Young diagrams must have a number of boxes which is divisible by~$d^2$.
	This can be seen for instance by induction: $(\pi_1 \ \text{and} \ \pi_\mathrm{ad})$ are two representations made of $0$ and $d^2$ boxes respectively. Now consider a representation $\pi$ whose number of boxes is divisible by $d^2$;
	$\pi \otimes (\pi_1\oplus\pi_\mathrm{ad})$ is again a direct sum of representation divisible by $d^2$, since tensoring with the trivial one does nothing and tensoring with the adjoint adds $d^2$ boxes to the Young diagram of $\pi$. According to Young calculus only $d^2$ boxes can be cancelled at once. Hence, if the statement is true for $(\pi_1\oplus\pi_\mathrm{ad})^{\otimes k-1}$, then it holds for $(\pi_1\oplus\pi_\mathrm{ad})^{\otimes k}$.
	\\
	Indeed, all such representations~$\pi$ other than the trivial and the adjoint one satisfy $c_2(\pi)>1$ as we will show in Lemma~\ref{lem:Young} below.
	\proofend
\end{proof}

Let $\lambda\coloneqq(\lambda_1,\dots,\lambda_{N-1})$ with $\lambda_i\in\mathbb{N}_0$ denote the Dynkin label of an irreducible representation~$\pi$ of~$\su(N)$.
The eigenvalue of the Casimir element in the irreducible representation~$\pi$ is
	\begin{equation} \label{eq:Casimir_Dynkin}
	c_2(\pi) =\frac{1}{2N} \sum_{i,j=1}^{N-1}(\lambda_i+2) (A^{-1})_{i, j} \lambda_j \ ,
	\end{equation}
where $A$ is the {Cartan} matrix of~$\su(N)$ \cite[\S 21.3]{FultonHarris}.
The  \emph{inverse} {Cartan} matrix is directly given by
	\begin{equation}
	(A^{-1})_{i,j} =\frac{1}{N}
	\begin{cases}
	i\mkern2mu (N-j) & \text{if $i\leq j$} \\
	j\mkern2mu (N-i) & \text{if $i>j$}
	\end{cases},
	\end{equation}
	% The Cartan matrix of $\su(N)$ is explicitly spelled out in Fulton-Harris, and the inverse Cartan matrix follows from that.
%	\MK{reference?}
and is symmetric.
We now show the following lemma.

\begin{lemma}[Young diagrams]
    \label{lem:Young}
    Let $N>2$ and $\pi$ be an irreducible representation of~$\su(N)$ such that the number of boxes in its Young diagram is divisible by~$N$.
    If $\pi$ is  \emph{not} isomorphic to the trivial or adjoint representation, then $c_2(\pi)>1$.
\end{lemma}

\begin{proof}
	First observe that we can immediately rule out all irreducible representations whose Young
	diagrams consist of a single column because the maximal column height for~$\su(N)$ is $N-1$
	(i.e. Dynkin labels having a single entry $1$ and 0 everywhere else).
	In the following we will analyse the growth behaviour of the quadratic form~\eqref{eq:Casimir_Dynkin} as
	 we move from one irreducible representation (i.e. Dynkin label) to the next one.
	
	It will turn out very helpful to know the column sums of the inverse Cartan matrix~$A^{-1}$.
	Clearly, the sum of the first (or equally the last) column is $(N-1)/2$.
	The sum of any other column is strictly greater than this value.
	Indeed, pick a column~$j$ and denote its sum by~$a_j$.
	One can easily convince oneself that $a_j=j\mkern2mu(N-j)/2$	.
	
	Now we compare the quadratic Casimir eigenvalues of different irreducible representations, i.e. Dynkin labels~$\lambda$.	
	As it turns out, adding~$1$ to any component of any Dynkin label~$\lambda$ always increases this eigenvalue at least by almost~$1/2$,
	\begin{equation}
	    c_2(\lambda+
	        e_i)-
	    c_2(\lambda)
	    \geq
	     \frac{N^2-
	           1}
	          {2
	           N^2}
	    \eqqcolon
	     \Delta_N.
	\end{equation}
	Here $e_i$ is the $i$-th canonical basis vector of~$\mathbb{R}^{N-1}$.
	So, starting from the trivial representation with $c_2(0,0,\dots,0)=0$ we immediately obtain the crude lower bound
	\begin{equation}
	    c_2(\lambda_1,
	        \dots,
	        \lambda_{N-1})
	    \geq
	     \Delta_N
	     \sum_{i
	           =1}^{N-
	                1}
	     \lambda_i
	    =\Delta_N
	     \norm{\lambda}_1.
	\end{equation}
	Observe that $2\Delta_N<1<3\Delta_N$.
	Thus we are guaranteed to obtain a quadratic Casimir eigenvalue strictly greater than~$1$ whenever we add at least three arbitrary columns to the (empty!) Young diagram of the trivial representation.

	This leaves us with those irreducible representations whose Young diagrams have exactly two columns, i.e. with the Dynkin labels $(1,1,0,\dots,0)$, $(2,0,0,\dots,0)$ and all permutations thereof.
	As is well known (and can be checked easily with the explicit formula below) the quadratic Casimir eigenvalue of the adjoint representation $(1,0,\dots,0,1)$ is exactly~$1$.
	We would like to show that any other placement of the two ones yields a strictly greater eigenvalue.
	Suppose these occur in positions $1\leq\alpha<\beta<N$.
	Then,
	\begin{equation}
	    \begin{split}
	        c_2(\lambda)
	        & =\frac{1}
	                {2
	                 N}
	           \bigl((A^{-1})_{\alpha,
	                      \alpha}+
	                 2
	                 (A^{-1})_{\alpha,
	                           \beta}+
	                 (A^{-1})_{\beta,
	                           \beta}+
	                 2
	                 a_\alpha+
	                 2
	                 a_\beta\bigr) \\
	        & \geq
	           \frac{1}
	                {2
	                 N}
	           \bigl((A^{-1})_{1,
	                           1}+
	                 2
	                 (A^{-1})_{1,
	                           N-
	                           1}+
	                 (A^{-1})_{N-
	                           1,
	                           N-
	                           1}+
	                 2
	                 a_1+
	                 2
	                 a_{N-
	                    1}\bigr) \\
	        & =c_2(1,
	               0,
	               \dots,
	               0,
	               1) \\
	        & =1.
	    \end{split}
	\end{equation}
	It is easy to see that this inequality turns into a strict one if either of the two ones is  \emph{not} at the first or last position.
	Finally consider a Dynkin label~$\lambda$ with a single non-vanishing component~$\lambda_\alpha=2$ at position~$\alpha$ (i.e., a Young diagram with exactly two columns of height~$\alpha$),
	\begin{equation}
	    c_2(\lambda)
	    =\frac{2}
	          {N}
	     \bigl((A^{-1})_{\alpha,
	                     \alpha}+
	           a_\alpha\bigr)
	    =\frac{N+
	           2}
	          {N^2}\mkern2mu
	     \alpha\mkern2mu
	     (N-
	      \alpha).
	\end{equation}
	From the global minimum of the quadratic function $\alpha\mkern2mu(N-\alpha)$ we easily obtain the lower bound
	\begin{equation}
	    c_2(\lambda)
	    \geq
	     1+
	     \frac{N-
	           2}
	          {N^2}
	\end{equation}
	and thus $c_2(\lambda)>1$ for all $N>2$ as claimed.
	\proofend\end{proof}

\subsection{Hamiltonian driving}
\label{sec:driving}
We now show that a time constant part in a stochastic Hamiltonian cannot affect the gap of the $k$-th moment operator.

\begin{lemma}[Hamiltonian driving]
	\label{lem:driving}
Let $M_T^k$ be the $k$-th moment operator \eqref{eq:k-moment} of a universal Brownian motion with increments $\Theta_{\Delta t}$ as in eq.~\eqref{eq:pseudo_derivative}.
Write $\Theta_{\Delta t}$ as
\begin{equation}\label{eq:H0Fdr}
  \Theta_{\Delta t}= -\i \, H_0 + F_{\Delta t}\, ,
\end{equation}
where $-\i\, H_0$ and $F_{\Delta t}$ are its anti-Hermitian time constant and fluctating parts, respectively, with
\begin{align*}
F_{\Delta t} = \sum_\mu B_\mu \, \xi^\mu_{\Delta t},
\qquad B_\mu^\dagger = -B_\mu \, ,
\qquad
\EE[\xi^\mu_{\Delta t}] = 0\, ,
\qquad \text{and}\qquad
\EE[\xi^\mu_{\Delta t} \, \xi^\nu_{\Delta t}] = - \frac{a}{{\Delta t}} \, \delta_{\mu,\nu} \, .
\end{align*}
Let $\tilde M_T^k$ be defined similarly but without driving, i.e., with $H_0 = 0$.
Then $\tilde M_T^k$ and $M_T^k$ have the same gap, i.e.,
\begin{equation}
\norm{\tilde M_T^k-M^k_{\Haar}}_\infty
=
\norm{M_T^k-M^k_{\Haar}}_\infty .
\end{equation}
\end{lemma}

\begin{proof}
Lemma~\ref{lemma:k-copy_gapped} implies that the gap of
$M_{\Delta t}^{k}$
is
$\norm{M_{\Delta t}^{k} -M^k_{\Haar}}_\infty$.
Hence,
\begin{equation}
\norm{M_T^{k}-M^k_{\Haar}}_\infty
=
  \lim_{\Delta t \to 0}
  \norm{	
    M_{\Delta t}^{k}
  -M^k_{\Haar}}_\infty^{T/\Delta t}
\end{equation}
is the gap of $M_T^{k}$ and, similarly, for $\tilde M_T^k$.

Using the connection between Brownian motion and its increments \eqref{eq:def:U_t} and a Trotter-Suzuki approximation we obtain
\begin{align}
  M_{\Delta t}^{k}
  &=
  \EE\left[
      \exp\{ \pi_{k,k}(-\i H_0 + F_{\Delta t})\, \Delta t\}
    \right] + O(\Delta t^2) \nonumber
  \\
  &=
  \EE\left[
      \exp\{ \pi_{k,k}(F_{\Delta t})\, \Delta t\}\right]\,
      \exp\{ \pi_{k,k}(-\i H_0)\, \Delta t\}
      + O(\Delta t^2)\nonumber
   \\
   &=
   \tilde M_{\Delta t}^{k} \exp\{ \pi_{k,k}(-\i H_0)\, \Delta t\}
      + O(\Delta t^2) \, .
\end{align}
As $\exp\{ \pi_{k,k}(-\i H_0)\, \Delta t\}$ is a fixed unitary, up to an error of order $O(\Delta t^2)$, the gap of
$M_{\Delta t}^{k}$ and
$\tilde M_{\Delta t}^{k}$
are the same.
This finishes the proof.
\proofend
\end{proof}

\subsection{More general interaction graphs}\label{sec:graphs}
The generator from Lemma~\ref{lem:generator} of the $k$-th moment operator of the unitary Brownian motion inherits the locality structure from the increments \eqref{eq:stochastic_local_H}.
Hence, it can be written as
\begin{equation}\label{eq:Gk-E}
	G^k = \sum_{e\in E} g_{k}^{(e)} \ ,
\end{equation}
where $G^k$ is the generator associated to $\Theta_{\Delta t}$ and $g_{k}^{(e)}$ to $\theta^{(e)}_{\Delta t}$ according to eq.~\eqref{eq:loc_g}.
Presumably, among all connected graphs, the gap of $G^k$ could have a minimum for $1D$ nearest neighbour graphs. Here, we show that adding edges to this graph can only increase the gap, which can only lead to a faster mixing in Theorem~\ref{thm_TPE}.

In the following lemma, the spectral gap $\Delta(G)$ of an operator $G$ is the difference of the second smallest and smallest singular value.

\begin{lemma}[The spectral gap of the generator is concave]
	\label{lem:CompleteGraph}
	Let $(G_i)_{i}$ be a finite set of negative semidefinite and Hermitian operators with common non-trivial kernel and $p$ be a probability vector.
	Then
	\begin{equation}
	  \Delta\left(\sum_{i} p_i \, G_i\right)  \geq \sum_i p_i \, \Delta(G_i) \, .
	\end{equation}
\end{lemma}

This lemma implies that the gap of the generator \eqref{eq:Gk-E} can only become smaller when one removes edges from $E$, while keeping $E$ connected.
Hence, the gap in the case of a one dimensional graph can also only be smaller as the gap in case of a complete graph.

\begin{proof}
Let $K$ denote the common kernel of $(G_i)_{i}$. Then it is also the kernel of any operator in the convex hull of $(G_i)_{i}$.
The gap of $G_i$ is the smallest singular value of $G_i$ restricted to the orthogonal complement of $K$ and similarly for $G\coloneqq \sum_{i} p_i \, G_i$.
Hence, it is enough to show that the smallest singular value as the function
\begin{equation}
  G \mapsto \min_{\braket{ x| x} =1} |\bra{x} G \ket{x}|
\end{equation}
is concave. But this follows from the smallest singular value being the minimum of the linear functions $G \mapsto \bra{x} G \ket{x}$.
\proofend
\end{proof}

\begin{remark}[Frustration free Hamiltonians]
  The same argument applies when the operators are all positive semidefinite.
  Hence, the gap of frustration free Hamiltonians, as considered in  ref.~\cite{BrandaoHarrowHorodecki}, is also a concave function, i.e., can only increase under taking convex combinations.
\end{remark}

\subsection{Example: White noise in the Pauli basis}
\label{xmp:white_noise}
We conclude the discussion on approximate unitary designs with an example involving the specific setting in eqs.~\eqref{eq:H_Delta_t} and \eqref{eq:covariance_xi}, and see that the choice of the Pauli matrices as a basis precisely matches, under the representation theoretic approach, the assumption on the covariance for the variables $\xi$.

Consider $n=2$ qubits (thus $N=4$) and the Hamiltonian increments
	\begin{equation}
	\Theta_{\Delta t}
	\coloneqq
	-\mathrm{i} \sum_{\alpha, \beta=0}^{3} {\sigma_{\alpha} \otimes \sigma_{\beta}\ \xi^{( \alpha, \beta)}_{\Delta t}} \, ,
	\end{equation}	
	where $ \xi_{\Delta t}^{( \alpha, \beta)}$ are i.i.d. real random variables with zero mean and covariance
	\begin{equation}
	    \label{eq:covariance_xi_example}
	    \cov[\xi_{\Delta
	              t}^{(\alpha,
	                   \beta)}
	         \xi_{\Delta
	              t}^{(\alpha',
	                   \beta')}]
	    =\delta_{\alpha,
	             \alpha'}
	     \delta_{\beta,
	             \beta'}
	     \frac{1}
	          {\Delta
	           t} \ , \qquad \forall \alpha,\beta.
	\end{equation}
	Leaving out the term $\sigma_0\otimes\sigma_0\,  \xi_{\Delta t}^{( 0, 0)}$ we can easily restrict~$\Theta_{\Delta t}$ to its traceless part
	\begin{equation}
	\Theta_{0,\Delta t}
	=
	\sum_{\mu=1}^{15} \tau_\mu \, \xi_{\Delta t}^\mu \ ,
	\end{equation}	
	where we defined the anti-Hermitian operators $\tau_\mu\coloneqq-\mathrm{i}\, \sigma_{\mu_1}\otimes\sigma_{\mu_2}$ so that $\set{\tau_1,\tau_2,\dots,\tau_{15}}=\set{\tau_{(0,1)},\tau_{(0,2)},\dots,\tau_{(3,3)}}$ form a basis of the fundamental representation of~$\su(4)$.
	From eq.~\eqref{eq:Killing_tensor_defining_rep}
	we compute the Killing metric tensor \eqref{eq:def_kappa} with respect to this basis as
	\begin{equation}
	    \kappa_{\mu,
	            \nu}
	    =-
	     8
	     \Tr(\tau_\mu^\dagger
	         \tau_\nu)
	    =-
	     32
	     \delta_{\mu,
	             \nu}.
	\end{equation}
From eq.~\eqref{eq:noise_cov} and the assumption in eq.~\eqref{eq:covariance_xi_example} immediately follows $a=32$. Observation~\ref{lem:casimir} tells us then
$
g_2=-16\, C(\pi_{2,2})=
$
and hence the second moment operator $M_{\SLH(\Delta t)}^{k=2}$ has a gap of $16\Delta t$, matching eq.~\eqref{eq:Gershgorin_bound} in the decoupling section.

%====================================================================================
%====================================================================================
%====================================================================================
\section{Decoupling with stochastic Hamiltonian time evolution}\label{sec:Decoupling_poof}
%====================================================================================
%====================================================================================
%====================================================================================
The section is devoted to the proof Theorem~\ref{thm_4.1}. To show our result, we consider a fluctuating Hamiltonian on a complete graph whose increments are given in eq.~\eqref{eq:H_Delta_t}, in the limit of $\Delta t \rightarrow 0$.
As already mentioned, this result is implies Theorem~\ref{Decoupling_main} by application of the same proof technique used for the random quantum circuit case in ref.~\cite{Decoupling}.

First, we analyse how the support size of an initial Pauli string evolves during the process, then we observe how the qubits are made invariant under relabelling of the Pauli elements; this, together with the permutation invariance condition, leads to the desired result.
Decoupling of an arbitrary $n$-qubit system $A$ is mainly described by the second moment operator induced by the evolution. The expansion coefficients in the Pauli basis are given in eq.~\eqref{def:Pauli_coefficient}.
We recall that, since the Brownian motion on $\U(2^n)$ is Markovian, the second moment operator at time $T$ on $X$ is given by concatenating $T/ \Delta t$ times the operator $M_{n,\SLH (\Delta t)}^{k=2}$, i.e.,
\begin{align}
M_{n,\SLH (T)}^{k=2} (X)
&=
\lim_{\Delta t \rightarrow 0} \underbrace{M_{n,\SLH (\Delta t)}^{k=2} \circ \dots \circ M_{n,\SLH (\Delta t)}^{k=2}}_{T/\Delta t \text{ times}} \, (X) \\
&\eqqcolon
\lim_{\Delta t \rightarrow 0} \bigcirc_{s=1}^{T/ \Delta t} M_{n,\SLH (\Delta t)}^{k=2}\, (X) \ .
\end{align}
Note that, since the Hamiltonian in eq.~\eqref{eq:H_Delta_t} generating Brownian motion is dependent on system size, we must include an additional subscript.
	
In Taylor approximation, up to an error $\O(\Delta t^2)$, $M_{n,\SLH (\Delta t)}^{k=2}$  results from the sum of two-qubit moment operators
acting on any possible qubit pair $j,k$, i.e.
\begin{equation}\label{eq:loc_structure_moment_operator}
M_{n,\SLH (\Delta t)}^{k=2}
=
\frac{2}{n(n-1)} \sum_{j<k} \left(M_{2,\SLH (\Delta t)}^{k=2}\right)^{j,k} +O(\Delta t ^2) \ .
\end{equation}
This can be seen through calculations analogous to the ones from eqs.~\eqref{generator_global_moment}-\eqref{eq:derive_g_2}.
We can hence interpret this process as a qubit pair being uniformly randomly chosen at every time step $(\ell-1) \Delta t$ and a two-qubit unitary $U_{2,\ell,\Delta t} \coloneqq \exp\{-\i\,  H_{2,\ell,\Delta t}\, \Delta t \}$ being applied.
Therefore, in the following section we first consider the restricted two-qubit case, which provides useful results and insights to be used for the investigation of the general case with $n$ qubits.

%===========================================
\subsection{Two-qubit analysis of the second moment operator}
\label{sec:two-qubit}
%===========================================

Considering a two-qubit system, here we would like to understand the evolution of $M_{2,\SLH(T)}^{k=2}$ through $M_{2,\SLH(\Delta t)}^{k=2}$ and show the following lemma, which is compatible to the analysis of the local gap discussed in the previous section (as showed in Example~\ref{xmp:white_noise}) .

\begin{lemma}[Two-qubit case]\label{prop:two_qubit_case}
	Then the local second moment operator associated to the Hamiltonian increments \eqref{eq:H_Delta_t} converges exponentially to the second moment operator of the uniform distribution, i.e.
	\begin{equation}
	\norm{ M_{2,\SLH (T)}^{k=2} - M_{2,\Haar}^{k=2}}_\infty \leq \e^{-16\,T} .
	\end{equation}
\end{lemma}

\begin{proof}
	To prove the convergence rate, we want to express $M_{n,\SLH (\Delta t)}^{k=2}$ in terms of the Pauli basis and compute the gap.
	We can see directly that the identity on $4$ qubits is an eigenvector with unit eigenvalue
	\begin{equation}
	M_{2,\SLH(\Delta t)}^{k=2} (\1_{4})
	=
	\E\left[U_{2,\ell, \Delta t}^{\otimes 2}\1_{4} (U_{2,\ell, \Delta t}^{\dagger})^{\otimes 2}\right]=\1_{4}.
	\end{equation}
	We then observe the unitary evolution acting on a Pauli element $\sigma_{\mu} \otimes \sigma_{\nu}$, with $\mu,\ \nu \in \{0,1,2,3\}^2$ and calculate its expectation with a Taylor expansion for the unitary, taking into account terms with leading order in $\Delta t$ (and omitting subscripts for $H$),
	\begin{align}
	M_{2, \SLH(\Delta t)}^{k=2} (\sigma_{\mu} \otimes \sigma_{\nu})
	&=
	\E\left[U_{2,\ell, \Delta t}\ (\sigma_{\mu_1} \otimes \sigma_{\mu_2}) \ U_{2, \ell,\Delta t}^{\dagger} \otimes U_{2, \ell,\Delta t} \ (\sigma_{\nu_1} \otimes \sigma_{\nu_2}) \ U_{2, \ell,\Delta t}^{\dagger} \right] \label{Taylor_initial}
	\\ &=
	\E\left[ \left(\1_2-\i H \Delta t -\frac{1}{2} H^2 \Delta t^2\right)\ (\sigma_{\mu_1} \otimes \sigma_{\mu_2})\ \left(\1_2+\i H \Delta t -\frac{1}{2} H^2 \Delta t^2 \right) \right.\nonumber
	\\ & \phantom{={}}
	\otimes \left.\left(\1_2-\i H \Delta t -\frac{1}{2} H^2 \Delta t^2\right)\ (\sigma_{\nu_1} \otimes \sigma_{\nu_2})\ \left(\1_2+\i H \Delta t -\frac{1}{2} H^2 \Delta t^2\right) \right]\nonumber \\
	&+ \O(\Delta t^2).\nonumber
	\end{align}
	We now recall that the $\xi$ white noise variables are i.i.d. with zero mean and covariance as in eq.~\eqref{eq:covariance_xi}. Considering only the non-vanishing linear terms in $\Delta t$ in the expectation, we have
	\begin{align}
	M_{2, \SLH(\Delta t)}^{k=2} (\sigma_{\mu} \otimes \sigma_{\nu})
	&=
	\sigma_{\mu} \otimes \sigma_{\nu} + \Delta t^2 \E\bigl[H\sigma_{\mu}H \otimes \sigma_{\nu} + \sigma_{\mu} \otimes H\sigma_{\nu}H\bigr] \nonumber\\
	&-
	\frac{\Delta t^2}{2}\E\left[H^2\sigma_{\mu} \otimes \sigma_{\nu} +\sigma_{\mu} H^2 \otimes \sigma_{\nu} + \sigma_{\mu} \otimes H^2 \sigma_{\nu}+ \sigma_{\mu} \otimes \sigma_{\nu} H^2\right]\nonumber\\
	&
	-\Delta t^2 \E\bigl[[H,\sigma_{\mu}] \otimes [H,\sigma_{\nu}] \bigr] + \O(\Delta t^2) \label{Taylor_final}.
	\end{align}
	Let us consider the second term, in particular
	\begin{equation}
	\begin{array}{ll}
	\E\left[H\sigma_{\mu}H \otimes \sigma_{\nu} \right] =\frac{1}{\Delta t} \left(\sum_{\alpha, \beta}(\sigma_{\alpha} \otimes \sigma_{\beta})(\sigma_{\mu_1} \otimes \sigma_{\mu_2})(\sigma_{\alpha} \otimes \sigma_{\beta}) \right) \otimes (\sigma_{\nu_1} \otimes \sigma_{\nu_2}).
	\end{array}\end{equation}
	If $\mu=0$, then
	\begin{equation}
	\E\left[H\, \1_2 \,H \otimes \sigma_{\nu} \right] =\E\left[H^2 \otimes \sigma_{\nu}\right]=\frac{16}{\Delta t}\ \1_2\otimes \sigma_{\nu}.
	\end{equation}	
	Otherwise, for $\mu \neq 0$, at least one among $\mu_1$ and $\mu_2$ is not 0. Let us assume $\mu_1 \neq 0$. Then, $\forall \beta$, $\sigma_{\alpha}\sigma_{\mu_1}\sigma_{\alpha} \otimes \sigma_{\beta}\sigma_{\mu_2}\sigma_{\beta}$ equals $\sigma_{\mu_1} \otimes \sigma_{\beta}\sigma_{\mu_2} \sigma_{\beta}$ for $\alpha=0,\mu_1$ and $-\sigma_{\mu_1} \otimes \sigma_{\beta}\sigma_{\mu_2}\sigma_{\beta}$ for the other two indices of $\alpha$. Thus, summing over $\alpha$ gives 0. The same applies for $\mu_1$ arbitrary, $\mu_2 \neq 0$.
	We conclude that the second term in the expression for $	M_{2, \SLH(\Delta t)}^{k=2}$ vanishes if both $\mu$ and $\nu$ are different from $\left\{0,0\right\}$.

	Now we look at the first part of the third term and we get that
	\begin{equation}
	\E\left[H^2\sigma_{\mu} \otimes \sigma_{\nu}\right]=\frac{1}{\Delta t}\sum_{\alpha, \beta}{(\sigma_{\alpha} \otimes \sigma_{\beta})^2 \sigma_{\mu} \otimes \sigma_{\nu}}=\frac{16}{\Delta t} \ \sigma_{\mu} \otimes \sigma_{\nu} .
	\end{equation}
	Hence, keeping terms to leading order in $\Delta t$ we have
	\begin{align}\label{eq:2mom1}
	M_{2, \SLH(\Delta t)}^{k=2}  (\sigma_{\mu} \otimes \sigma_{\nu} ) &= (1-32 \Delta t) \sigma_{\mu} \otimes \sigma_{\nu} -\Delta t^2 \E\left[[H,\sigma_{\mu}] \otimes [H,\sigma_{\nu}] \right]\\
	&=(1-32 \Delta t) \sigma_{\mu} \otimes \sigma_{\nu} -\Delta t \sum_{\alpha,\beta}{[\sigma_{\alpha}\otimes \sigma_{\beta},\sigma_{\mu_1}\otimes \sigma_{\mu_2}] \otimes [\sigma_{\alpha}\otimes \sigma_{\beta},\sigma_{\nu_1}\otimes \sigma_{\nu_2}]},
	\label{eq:2mom2}
	\end{align}
	when both $\mu$ and $\nu$ are different from $\left\{0,0\right\}$, and conversely
	\begin{align}\label{1sig1}
M_{2, \SLH(\Delta t)}^{k=2}  (\1_2 \otimes \sigma_{\nu} )
	&=
	(1-16 \Delta t)\1_2 \otimes \sigma_{\nu},
	\\
	M_{2, \SLH(\Delta t)}^{k=2} (\sigma_{\mu} \otimes \1_2 )
      &=
	(1-16 \Delta t)\sigma_{\mu} \otimes \1_2.\label{1sig2}
	\end{align}
	We now divide the set of all possible strings $\sigma_{\mu} \otimes \sigma_{\nu}$ in three parts: the
	identity $\1_{4}$, the set of strings of the form $\sigma_{\mu} \otimes \sigma_{\mu}$, and all remaining strings of the
	form $\sigma_{\mu}\otimes \sigma_{\nu}$ with $\mu \neq \nu$.
	We can then make use of the matrix representation of the
	operator $M_{2, \SLH(\Delta t)}^{k=2}$ as a matrix with respect to Pauli basis,
	which gives
	\begin{equation}
		M_{2, \SLH(\Delta t)}^{k=2} =
		\begin{pmatrix}
		1& & \\
		& A &  \\
		& & B
		\end{pmatrix} \ ,
		\end{equation}
		where $A$ is a $15 \times 15$ matrix related to the set of $\sigma_{\mu} \otimes \sigma_{\mu}$ elements (without the identity $\1_{16}$) and $B$ is a $240 \times 240$ matrix for $\sigma_{\mu} \otimes \sigma_{\nu}$ elements. The detailed proof of this finding is laid out in the
		separate subsequent Lemma
	\ref{lem:Hermitian}.
%	We then formulate the following Lemma.
	
%	\MK{I think it can be rather confusing to have another lemma in a proof. can this be put elsewhere? Especially as the proof is quite long already.}\eo{[I would prefer to keep it here, since is really short and sticks well with the reading flow in my opinion]}
%
% Compromise: The result is stated here, but the Lemma presented later. I prefer this option for clarity.

	We now consider the matrix $A$; we compute the action of $M_{2, \SLH(\Delta t)}^{k=2} $ over all possible $\sigma_{\mu} \otimes \sigma_{\mu}$ and look for eigenvalues.
	We obtain a non-degenerate eigenvalue $1$ whose eigenvector is the uniform sum over all non-identity Pauli matrices
	\begin{equation}
	\mathds{F}=\frac{1}{15}\; \sum_{\gamma \neq 0} \sigma_{\gamma} \otimes \sigma_{\gamma} .
	\end{equation}
	We then have a $9$-fold degenerate eigenvalue $1-40 \Delta t$ and a $5$-fold degenerate eigenvalue $1-24\Delta t$. We are free to bound all these eigenvalues with $1-16\Delta t$.
	We now deal with the action of the second moment operator on terms of the form $\sigma_{\mu}\otimes \sigma_{\nu}$ with $\mu,\nu \neq 0$ and $\mu \neq \nu$. Only four choices of $\sigma_{\alpha} \otimes \sigma_{\beta}$ do not commute for a given pair $\mu,\nu$, i.e.:
	\begin{equation}
	M_{2, \SLH(\Delta t)}^{k=2} (\sigma_{\mu} \otimes \sigma_{\nu})=(1-32 \Delta t)\sigma_{\mu} \otimes \sigma_{\nu} - 4 \Delta t \{ \pm \sigma_{\gamma_1} \otimes \sigma_{d_1} \pm \sigma_{\gamma_2} \otimes \sigma_{d_2} \pm \sigma_{\gamma_3} \otimes \sigma_{d_3} \pm\sigma_{\gamma_4} \otimes \sigma_{d_4} \}
	\end{equation}
	with $\gamma_i \neq d_i$, for each $\sigma_{\mu} \otimes \sigma_{\nu}$. This means that each column of the matrix $B$ has one entry $(1-32 \Delta t)$ (in the diagonal element) and four entries $\pm 4 \Delta t$, and $0$ otherwise. Hence,
	\begin{equation}
	\|B\|_1 = \max_{j}\sum_{i} |a_{i,j}|=1-16 \Delta t .
	\end{equation}
	By the Gershgorin circle theorem, and taking also into account \eqref{1sig1} and \eqref {1sig2}, we can upper bound the highest eigenvalue of $B$ with $1-16\Delta t$.
	For a single time step, the two-qubit second moment operator can be upper bounded by the following diagonal matrix
	\begin{equation}\label{eq:Gershgorin_bound}
	M_{2, \SLH(\Delta t)}^{k=2} \leq
	\begin{pmatrix}
	1 &   & & & \\
	& 1 & & & \\
	& & 1-16\Delta t &  & \\
	& & & \ddots   & \\
	& & & &  1-16 \Delta t
	\end{pmatrix},
	\end{equation}
	where we recall that the $2$-fold degenerate eigenvalue $1$ corresponds to the identity and $\omega$.
	\proofend\end{proof}

	%=============================================
	% Lemma
	%=============================================
	\begin{lemma}[Local second moment operator]\label{lem:Hermitian}
		$M_{2, \SLH(\Delta t)}^{k=2}$ is Hermitian, maps elements of the set of strings of the form $\sigma_{\mu} \otimes \sigma_{\mu}$ to a linear combination of elements of the same set and elements of the set of strings of the form $\sigma_{\mu}\otimes \sigma_{\nu}$ with $\mu \neq \nu$ again to a linear combination of elements of the same set, such that there is no mixing between the two sets.
		Hence, we can represent the operator $M_{2, \SLH(\Delta t)}^{k=2}$ as a matrix with respect to Pauli basis in the following form
		\begin{equation}
		M_{2, \SLH(\Delta t)}^{k=2} =
		\begin{pmatrix}
		1& & \\
		& A &  \\
		& & B
		\end{pmatrix} \ ,
		\end{equation}
		where $A$ is a $15 \times 15$ matrix related to the set of $\sigma_{\mu} \otimes \sigma_{\mu}$ elements (without the identity $\1_{16}$) and $B$ is a $240 \times 240$ matrix for $\sigma_{\mu} \otimes \sigma_{\nu}$ elements.
	\end{lemma}
	\begin{proof}
		From eq.~\eqref{eq:2mom1} and \eqref{eq:2mom2}, follows directly that $M_{2, \SLH(\Delta t)}^{k=2} $ is Hermitian.
		Moreover we see, again from eq.~\eqref{eq:2mom2}, that elements of the set $\sigma_{\mu} \otimes \sigma_{\mu}$ are mapped to a linear combination of elements of the same set.
		This, in addition to the fact that $M_{2, \SLH(\Delta t)}^{k=2} $ is Hermitian, implies that elements of the set $\sigma_{\mu}\otimes \sigma_{\nu}$ with $\mu \neq \nu$ are mapped again to a linear combination of elements of the same set.
		\proofend\end{proof}
	%============================================
	%End Lemma
	%============================================

Next, we make use of this analysis to understand how $n$-qubit Pauli strings evolve during the continuous-time process. In Appendix~\ref{app:stochastic_processes} we collect the most relevant mathematical tools used in the second part of this section. As already mentioned, the continuous-time random walk induced by the Hamiltonian increments can be interpreted as a sequence of jumps defining a discrete random walk spaced out by i.i.d.\ waiting times.

%=================================================================
\subsection{Markov chain analysis on weights}
%=================================================================

The proof strategy for Lemma \ref{thm_4.1} begins with the analysis of the evolution of the coefficients: we observe how the support size behaves during the process, inferring a probability that, for a given initial string $\sigma_\mu$ with support size $\ell$, after run time $T$ the string has support size $k$.
Conditioned on some specific event $E_W$ that we will discuss later, this probability can be upper bounded as
\begin{equation}\label{eq:bound_weight}
\mathbb{P}\left( \{T,\ell,k\} \ \big|\ E_W\right) \coloneqq \sum_{\left|\nu\right|=k}Q_{E_W}^T(\mu,\nu) \leq \binom{n}{k}3^k\frac{4^{\delta n}}{4^n-1} .
\end{equation}
Having a total of $\binom{n}{k} 3^k$ strings with support size $k$, we then show that almost all of them have the same probability.

Considering the analysis in the previous section on the two-qubit case and that, the local structure of $M_{n,\SLH (\Delta t)}^{k=2}$ given in eq.~\eqref{eq:loc_structure_moment_operator}
we introduce a Markov chain over the weights of the string similarly to ref.\ \cite{RandomCircuitsLow} (where this projected chain is called \emph{zero chain}).
The chain runs over the state space $\Omega=\{1,2,\dots,n\}$ and the transition probability from $\ell$ at time $t$ to $k$ at time $t+\Delta t$ is described by the matrix element
\begin{equation}
P(\ell,k) \coloneqq \sum_{\nu:\left|\nu\right|=k}\frac{1}{4^n} \Tr\left[ \sigma_\nu\otimes \sigma_\nu\, M_{n,\SLH (\Delta t)}^{k=2} (\sigma_\mu\otimes \sigma_\mu) \right]
\end{equation}
for any choice of $\mu$ with support size $\ell$.

\begin{lemma}[Transition matrix of the zero chain]\label{transition_matrix_P}
	The zero chain has transition matrix $P$ on state space $\Omega=\{1,2,\dots,n\}$,
	\begin{equation}\label{transmatrix}
	P(\ell,k)=\left \{
	\begin{array}{ll}
	1-\frac{16\ell(3n-2\ell-1)}{n(n-1)}\Delta t & k=\ell\\
	\frac{16\ell(\ell-1)}{n(n-1)}\Delta t & k=\ell-1\\
	\frac{48\ell(n-\ell)}{n(n-1)}\Delta t & k=\ell+1\\
	0 & \text{otherwise}
	\end{array}
	\right.
	\end{equation}
	for $1\leq x,y \leq n$.
\end{lemma}

\begin{proof}
	We consider the analysis of the two-qubit second moment operator in Section \ref{sec:two-qubit}. It is straightforward to note that, after application of $M_{n, \SLH(\Delta t)}^{k=2} $, the weight of the string can only vary by $1$ or stay the same.
	The weight decreases if a pair of two non-identity terms $\sigma \otimes \sigma$ is chosen and is transformed in a pair with one identity element (namely $\sigma\otimes \1$ or $\1\otimes \sigma$); there are in total four choices for $\sigma_{\alpha} \otimes \sigma_{\beta}$ which produce such a transition. According to the two-qubit case, the probability that one of these Pauli operators is chosen is $4 \cdot 4\Delta t= 16\Delta t$ and since the probability of choosing a pair with weight $2$ is
	${\ell(\ell-1)}/({n(n-1)})$, we have
	\begin{equation}
	P(\ell,\ell-1)=\frac{16\ell(\ell-1)}{n(n-1)}\Delta t \ .
	\end{equation}	
	The weight of the string can be increased if an identity term paired with a non-identity term is chosen (i.e., $\sigma \otimes\1$ or $\1\otimes\sigma$) and transformed into a pair of two non-identity terms $\sigma \otimes \sigma$. The probability of obtaining  such a result (conditioned on choosing such a pair) after application of the two-qubit second moment operator is $24\Delta t$, since there are in total 6 choices for $\sigma_{\alpha} \otimes \sigma_{\beta}$ to produce such a transition. Furthermore, the probability of choosing an identity and non-identity pair is given by
	${2\ell(n-\ell)}/({n(n-1)})$; hence
	\begin{equation}
	P(\ell,\ell+1)=\frac{48\ell(n-\ell)}{n(n-1)}\Delta t.
	\end{equation}	
	Finally, the probability of staying at the same weight is obtained by simply requiring the total probability to sum to unity.
	\proofend\end{proof}
It is therefore possible to reach each state of the chain, meaning that it is \emph{irreducible}. Moreover, the chain contains self loops, being hence \emph{aperiodic}. From these two properties follows that the chain is also \emph{ergodic},
thus converging to a unique stationary distribution.
\begin{lemma}[Stationary distribution of zero chain]\label{statdistr}
	The stationary distribution of the zero chain is
	\begin{equation}
	\omega_{0}(k)=\frac{3^{k}\binom{n}{k}}{4^n-1}.
	\end{equation}
\end{lemma}
\begin{proof}
	This follows from straightforward calculation.
	\proofend\end{proof}
The stationary distribution is actually analogous to the one of the chain induced by a random quantum circuit under the Haar measure (see ref.\ \cite[Lemma 5.3]{RandomCircuitsLow}).
Another crucial analogy is the exact equivalence of the \emph{accelerated chain} (i.e., the chain conditioned on moving) of the two different settings. This means that, when moving, the random walk on weights is identically biased for both random quantum circuits under Haar distribution and the stochastic Hamiltonian process. From the description of Montroll and Weiss, the jumps of the random quantum circuit are contained in the fluctuating Hamiltonian evolution, spaced out by i.i.d. waiting times.
Concretely, the accelerated chain is given by
\begin{equation}\label{0chain_transition_matrix}
{P_{\rm accel}}(\ell,k)=\left \{
\begin{array}{ll}
0 & k=\ell\\
\frac{\ell-1}{3n-2\ell-1}& k=\ell-1\\
\frac{3(n-\ell)}{3n-2\ell-1}& k=\ell+1\\
0 & \text{otherwise.}
\end{array}
\right.
\end{equation}
With these analogies, we can prove the next theorem using results from the proof of ref.\ \cite[Theorem 4.2]{Decoupling}. We should take care of the parts of the proof involving the waiting time, because it is where the two walks differ. We will also deal with the permutation invariance property in a more precise and explicit way.
Now, we reformulate the result for the continuous-time case.

%==================================================
\begin{lemma}[Mixing condition on support size]\label{lemma:mixing_main}
	%==================================================
	Let $P$ be the Markov chain transition matrix defined in Lemma \ref{transition_matrix_P}. For any constants $\delta \in(0,1/16), \eta \in (0,1)$ there exists a constant $\varsigma>0$ such that for $T \geq \varsigma \, n \log^2 n$ and all integers $1\leq \ell \leq n$ and $1 \leq k \leq n$, we have for large enough $n$
	\begin{equation}\label{eq_of_thm_4.2}
	\mathbb{P}\left(\left\{T,\ell,k\right\}\right)
	=
	\sum_{\nu:\left|\nu\right|=k}Q^T(\ell, k) \leq \binom{n}{k} 3^k \frac{4^{\delta n}}{4^n-1}+\frac{1}{(3-\eta)^\ell\binom{n}{\ell}}\frac{1}{{\rm poly}(n)} \ ,
	\end{equation}
	where $\left\{T,\ell,k\right\}$ is the event that an initial Pauli string with support size $\ell$, after a run time $T$, has weight equal to $k$.
\end{lemma}
%==============================================================
%==============================================================

\begin{proof}
	We start by defining the following points,
	\begin{equation}
	r_{-}\coloneqq \left(\frac{3}{4}-\delta\right)n \quad and \quad r_{+}\coloneqq \left(\frac{3}{4}+\delta\right)n .
	\end{equation}
	Then, considering ref.\ \cite[eq.~(20)]{Decoupling}, it follows that for an initial weight of $\ell \in \left[r_{-},r_{+}\right]$
	\begin{equation}
	\mathbb{P}(\left\{ T,\ell,k \right\}) \leq \binom{n}{k} 3^k \frac{4^{\delta n}}{4^n-1}
	\end{equation}
	for any $T > 0$.

	To deal with the case $\ell \in \left[1,r_{-}\right)$, for random quantum circuits it has been shown that the probability that the interval $\left[r_{-},r_{+}\right]$ of the state space has been reached is very high for a number of gates $\O(n\log^2 n)$. Here we prove the same scaling result for the run time of the continuous-time process, that is, the total waiting time between the jumps can be bound with the following lemma.
	
	%===================================================================
	\begin{lemma}[Waiting time]\label{lemma:waiting_time}
		%===================================================================
		\begin{equation}
		\mathbb{P}(E_W^c) \coloneqq \mathbb{P}(W_{r_{-}}> \varsigma\,n \log^2 n)
		\leq
		\frac{1}{(3-\eta)^\ell\binom{n}{\ell}}\frac{1}{{\rm poly}(n)}
		\end{equation}
		for some sufficiently large $\varsigma$.
	\end{lemma}
	%===================================================================

	The proof of the lemma is postponed to Appendix~\ref{app:proof_Lemma_waiting_time} to help readability. The case that remains to be discussed is the one of an initial Pauli string with support size $\ell \in \left(r_{+},n\right]$ to reach $\left[r_{-},r_{+}\right]$; again the analysis is divided on accelerated steps and waiting times.
	Regarding the former, the probability of going backward is larger than the one of moving forward starting from point $z$ with
	\begin{align}
	P(z,z+1)&\overset{!}{=}P(z,z-1),\\
	\frac{z-1}{3n-2z-1}&=3\frac{n-z}{3n-2z-1},
	\end{align}
	from which follows that
	\begin{equation}
	z=\frac{3}{4}n+\frac{1}{4}.
	\end{equation}
	This means that for any $n>{1}/({4\delta})$ the probability of moving backward at each site of region $\left(r_{+},n\right]$ is at least $1/2+\epsilon$ for some $\epsilon>0$, and again using the argument for the case with $\ell<r_{-}$ the probability of not reaching $r_{+}$ in $S\leq s$ steps is upper bounded by an exponential decreasing function for  $s \geq \phi'n$ for sufficiently large $\phi'$.
	In this instance, all waiting times are stochastically dominated by parameter $p(3n/4)=12$, hence there is no necessity to define an event equivalent to $H$. For $S \leq s$ accelerated steps, using again a Markov's inequality, the bound on the total waiting time is exponentially decreasing in $s$ for a run time
	$W_{r_{+}}>({\log 2}/{6})s$.
	The proof of Lemma \ref{lemma:mixing_main} is then complete.
	\proofend\end{proof}

%=====================================================
%=====================================================
\subsection{From the zero chain to the full distribution}
%=====================================================
%=====================================================

Once the weight distribution has reached an equilibrium such that the condition in eq.~\eqref{eq:bound_weight} is fulfilled, we need to show that all Pauli strings sharing the same weight have a similar probability. To prove this, we need to show that almost all Pauli strings with the same support but different Pauli labels $\{1,2,3\}$ are equivalent in probability. This, together with the permutation invariance property assumed for the initial state, which is conserved during the whole stochastic Hamiltonian process, will bring us to the desired result.

Let $M$ be the Markov chain on the first $n$-qubits induced by $M_{n,\Delta t}^2$, and define an accelerated version as
\begin{equation}\label{accel_chain_2}
A \coloneqq \frac{1}{36\Delta t}(M-(1-36 \Delta t)\mathcal{I}) .
\end{equation}
If we  define an operator
\begin{equation}
R=\frac{2}{n(n-1)}\sum_{j<k}R_{j,k} \ ,
\end{equation}
where $R_{j,k}$ randomises one qubit site in the following way,
\begin{equation}\label{randomize}
R_{j,k}(\sigma_\mu^j\otimes\sigma_\nu^k)=\left \{
\begin{array}{ll}
\frac{1}{3} \sum_{\alpha={1,2,3}} \sigma_\alpha^j \otimes \1^k & \text{if } \mu\neq0,\nu=0,\\
& \\
\frac{1}{3} \sum_{\alpha={1,2,3}} \1^j \otimes \sigma_\alpha^k & \text{if } \mu=0,\nu\neq 0,\\
& \\
\frac{1}{6} \sum_{\alpha={1,2,3}} \sigma_\alpha^j \otimes \sigma_\nu^k + \frac{1}{6} \sum_{\alpha={1,2,3}}
\sigma_\mu^j \otimes \sigma_\alpha^k & \text{if } \mu\neq 0,\nu \neq 0,\\
& \\
\1^j\otimes \1^k & \text{if } \mu=\nu=0,
\end{array}
\right.
\end{equation}
then according to Section \ref{sec:two-qubit}, the accelerated chain can be written as
\begin{equation}\label{acc_chain}
A=\frac{1}{3}R+\frac{2}{3}L,
\end{equation}
where
\begin{equation}
L=\frac{2}{n(n-1)}\sum_{j<k}L_{j,k}
\end{equation}	
and
\begin{equation}\label{transposition_invariant_chain}
L_{j,k}(\sigma_\mu^j\otimes\sigma_\nu^k)=\left \{
\begin{array}{ll}
\frac{1}{6} \sum_{\alpha={1,2,3}} \sigma_{\mu+1}^j \otimes \sigma_{\alpha}^k+\frac{1}{6} \sum_{\alpha={1,2,3}} \sigma_{\mu+2}^j \otimes \sigma_{\alpha}^k & \text{if } \mu\neq0,\nu=0,\\
&\\
\frac{1}{6} \sum_{\alpha={1,2,3}} \sigma_{\mu}^j \otimes \sigma_{\nu+1}^k  + \frac{1}{6} \sum_{\alpha={1,2,3}} \sigma_{\mu}^j \otimes \sigma_{\nu+2}^k  & \text{if } \mu=0,\nu\neq 0,\\
&\\
\frac{1}{12}\left( \sigma_{\mu+1}^j \otimes \sigma_\nu^k +\sigma_{\mu+2}^j \otimes \sigma_\nu^k + \sigma_\mu^j \otimes \sigma_{\nu+1}^k +\sigma_{\mu}^j \otimes \sigma_{\nu+2}^k \right) & \\
+\frac{1}{6} \left( \sigma_{\mu+1}^j \otimes \1^k+ \sigma_{\mu+2}^j \otimes \1^k+ \1^j \otimes \sigma_{\nu+1}^k + \1^j \otimes \sigma_{\nu+2}^k \right) & \text{if } \mu\neq 0,\nu \neq 0 ,\\
&\\
\1^j\otimes \1^k & \text{if } \mu=\nu=0,
\end{array}
\right.
\end{equation}
with the notation $\sigma_{3+1}=\sigma_{2+2} = \sigma_{1}$ and $\sigma_{3+2}=\sigma_{2}$.
Note that $R$ does not produce any change in the weight or transpositions between identities and non-identity elements, it solely performs a local randomisation of the Pauli labels. This means that only the chain $L$ is responsible for the random walk on the weights.

We would like to upper bound the probability that more than $\beta n$ sites have not been randomised after $s$ steps of chain $R$ (we denote the complement of this event as $E_R$).
Knowing that there are $\binom{n}{\beta n}$ such regions, this is given by union bound
\begin{equation}\label{beta_event}
\mathbb{P}(E_{R}^c) \leq \binom{n}{\beta n} (1-\beta)^{s} \leq
2^{h(\beta) n} \; \e^{-\beta s},
\end{equation}
where $h:[0,1]\rightarrow [0,1]$ is the binary entropy function.
This probability can then be upper bounded by an arbitrary exponentially decreasing function in $n$ for some $s=\O(n)$.
Hence, to ensure that $s$ randomisations have been performed to fulfill the event $E_R$ with  sufficiently large probability, given eq.~\eqref{acc_chain} and by application of an Hoeffding's inequality follows that it is again sufficient to apply $\O(n)$ steps of the accelerated chain $A$.
Since the waiting time is dominated by an exponential distribution with parameter $36$, the bound on the probability for the waiting time of this process to exceed $W_R =\varsigma_R \, n$ can be bounded by an arbitrarily exponentially decreasing function in $n$ for a sufficiently large $\varsigma_R$ with the same argument used for the random walk on weights when starting from $\ell> r_+$.

%================================================================================
%================================================================================
%================================================================================

In conclusion, assuming that event $E_W$ and $E_{R}$ have been satisfied, we have for $\gamma<\gamma_0 \leq 1/2$ :
\begin{enumerate}
	%[label={\{\arabic*\}}]
	\item For strings $\nu$ with support size $k \leq\gamma_0\,n$,
	\begin{equation}
	Q^T(\mu,\nu)
	\leq
	\sum_{\left|\nu\right|=k}Q^T(\mu,\nu)
	\leq
	\binom{n}{\gamma_0 \, n}3^{\gamma_0 \, n}\frac{4^{\delta n}}{4^n-1} \leq 2^{n\, h(\gamma_0)}\,3^{\gamma_0\, n} \frac{4^{\delta n}}{4^n-1}.
	\end{equation}
	\item For strings $\nu$ with support size $k \geq (1-\gamma_0)n$, given event $E_{R}$ at least $(1-\beta)n$ sites of the support have been uniformly randomised, hence
	\begin{equation}
	Q^T(\mu,\nu)
	\leq
	\frac{1}{3^{k-\beta \, n}}\sum_{\left|\nu\right|=k}Q^T(\mu,\nu)
	\leq
	2^{n\, h(\gamma_0)}\,3^{\beta \, n} \frac{4^{\delta n}}{4^n-1}.
	\end{equation}
	\item For  strings $\nu$ with support size $\gamma_0 \, n < k=\kappa n < (1-\gamma_0)n$ such that $\kappa-\gamma_0=O(1)$ (otherwise, we can apply slightly modified versions of the bounds in the two previous cases),  given event $E_{R}$ at least $(1-\beta)n$ sites of the support have been uniformly randomised. In addition, if we assume the $\gamma$-permutation invariance property for the initial string $\sigma_\mu$, we obtain
	\begin{equation}
	Q^T(\mu,\nu)
	\leq
	\frac{1}{3^{k-\beta \, n}}\frac{1}{\binom{(1-\gamma)n}{k-\gamma n}} \sum_{\left|\nu\right|=k}Q^T(\mu,\nu)
	\leq
	3^{\beta \, n}\, \left[\frac{1}{\kappa-\gamma_0}\right]^{\gamma_0\, n}\; \frac{4^{\delta n}}{4^n-1}.
	\end{equation}
\end{enumerate}
Now, for an appropriate choice of $\beta$ and $\gamma_0$,
\begin{equation}
Q^T(\mu,\nu) \leq \frac{5^{\delta n}}{4^n-1}
\end{equation}
for all $\mu$ and $\nu$.

Also, having proven that there exists $\varsigma$ such that, for all $T \geq \varsigma n \, \log^2 n$, $\mathbb{P}(E_{R}^c)$ is bounded by an exponentially decreasing function in $n$ and
\begin{equation}
\mathbb{P}(E_W^c)\leq \frac{1}{(3-\eta)^\ell\binom{n}{\ell}}\frac{1}{{\rm poly}(n)} .
\end{equation}
Having proven that, if both event have been satisfied and the permutation invariance property is assumed, we have
\begin{equation}
Q^T(\mu,\nu) \leq \frac{5^{\delta n}}{4^n-1}
\end{equation}
for all $\mu$ and $\nu$, we conclude the proof for the main Lemma \ref{thm_4.1}.

As mentioned in the main result section, the decoupling Theorem holds for all states which are invariant with respect to any permutation on $(1-\gamma)n$ qubits, in the sense of Definition~\ref{3P}, and not only for Pauli strings taken singularly.
Consider a set of 
$\min \left\{ \binom{n-\gamma n}{\ell -\gamma n} ,\binom{n-\gamma n}{\ell } \right\} \leq b_{\ell,\gamma} \leq \max \left\{ \binom{n-\gamma n}{\ell -\gamma n},\binom{n-\gamma n}{\ell } \right\}$ 
Pauli strings $\set{\sigma_\mu}_\mu$ with support size $\ell$ which is invariant with respect to any of such permutations.
Assuming that the above events have been satisfied, at least the same number of qubits in the final Pauli strings $\set{\sigma_\nu}_\nu$ is invariant with respect to permutations since the stochastic evolution preserves this property. Hence, for the argument from the previous subsection, we have: 
\begin{equation}
\sum_{\mu } Q^T(\mu,\nu) \leq b_{\ell,\gamma} \,  \frac{5^{\delta n}}{4^n-1} 
\end{equation}
This, together with the fact that $\Tr[\sigma_\mu \rho]$ is the same for all strings related by these permutations, allows to apply the proof in ref.~\cite{Decoupling} for the decoupling Theorem for all density states $\rho$ composed by permutation invariant sets of Pauli strings.

\section{Conclusions and outlook}

In this work, we have investigated mixing properties of fluctuating local Hamiltonian evolutions, establishing a connection with random quantum circuits. The two settings
differ on the distribution over unitary group: in the random quantum circuits considered in other works two-qudit gates are chosen from the Haar measure or a fixed distribution. The discretised stochastic Hamiltonian is described by local terms weighted random coefficients also generating a gate set. However, the gate set depends on the discrisation which required an involved analysis of the gap of the local moment operators.
We show that scaling in the system size in order to obtain an approximate unitary $k$-design are compatible in the two settings: the total run time of the diffusion process provides a faster mixing time, by a factor of $n$, in comparison of a local random quantum circuit due to the larger number of interactions per time step, but the two scenarios display the same scaling when they interact with two qubits only at each step. In this way, we provide a unifying framework of random quantum processes.

In order to bound the gap of the local moment operator, we have made use of and further developed
tools from representation theory, significantly going beyond uses of representation theory in related contexts \cite{RandomHamiltonian,Plenio}.
With this, we analyse how quickly the diffusion on the unitary group induced by the local stochastic Hamiltonian mixes, where the local gap characterises the speed of the diffusion. The gap can be lower bounded by an expression which is entirely
independent of the number of copies $k$ of the system, which constitutes a possibly surprising result in its own right.

In the framework of a continuous-time random walk on weights induced by the stochastic Hamiltonian evolution, we prove a decoupling theorem with almost linear scaling in $n$, already shown to be valid for random quantum circuits. The exact correspondence between the accelerated steps of the walk derived from the random circuit and the jumps of the continuous-time random walk originating from the stochastic Hamiltonian is a strong element of similarity: we can consider the steps of the circuit as if they were dispersed within the continuous-time process and spaced out by i.i.d.\ waiting times. Again, a unifying picture is hence provided.

All these results allow us to unify in one single mathematical framework  random quantum processes in the
form of quantum circuits and continuous-time phenomena governed by time-fluctuating Hamiltonians.
This is of interest for both a pragmatic and application-oriented \cite{RandomCircuitsLow,Speedups,SuperpolySpeedup,BoutenHandel,BrandaoHarrowHorodecki,Pseudorandom,SzeDuTomRen13,Decoupling} as
well as  a conceptual point of view \cite{FastScrambling,FastScramblingConjecture,Scrambling}, indeed giving guidance on how fast
time-fluctating processes lead to mixing or ``fast scrambling''.

Given the close connection of fluctuating processes with classes of local dissipative processes,
we also gain new insights into the impact of dissipation to quantum many-body dynamics. Turning the logic of approximating the Haar measure upside down, this
work shows how dynamics can deviate from the uniform measure without affecting its mixing properties. It is the hope that the present
work stimulates further research on random quantum processes, both as far as the mathematical development and the exploration of its
implications are concerned.

\section{Acknowledgements}
We thank  M.\ Horodecki, D.\ Gross, H.\ Wojew\'odka and I.\ Roth for fruitful discussions and
acknowledge support from the EU (RAQUEL, AQuS), the DFG (CRC 183, EI 519/7-1), the Templeton Foundation 
(RQ-35601),
the ERC (TAQ), 
and the BMBF (Q.com).

%%% ----------------------------------------------------------------------------
%%% ------------------------------ Bibliography --------------------------------
%%% ----------------------------------------------------------------------------
\setlength{\bibsep}{2pt}
\renewcommand{\bibfont}{\small}

\bibliographystyle{is-unsrt}
%\bibliography{Randomness}

\appendixtitleon
\appendixtitletocon
\begin{appendices}

\section{(Young diagrams)}\label{app:Young_diagrams}

In order to study the decomposition of the mixed tensor representation $\pi_{k,k}$, we make use of  \emph{Young diagrams} for $\su(N)$. These are arrays of boxes arranged in $N-1$ left-justified rows whose length is non-increasing from top to bottom, each of them connected to an irreducible representation, e.g.,

\vspace{.35cm}

\begin{equation*}\ytableausetup{boxsize=0.85em, aligntableaux=center}
\ydiagram{1,1,1,1} \hspace{3cm} \ydiagram{3,2,1} \hspace{3cm}\ydiagram{2,2,2} \hspace{3cm} \ydiagram{5,2} \ .
\end{equation*}

\vspace{.35cm}

In particular, the following holds true.
\begin{itemize}
	\item The Young diagram of the fundamental representation is given by one single box $\ytableausetup{boxsize=0.85em, aligntableaux=center}
	\ydiagram{1}$\, .
	\item The trivial representation does not have any box; we can denote it by $\emptyset$.
	\item The adjoint representation is given by a column of $N-1$ boxes and a second column made of a single box. For example, the adjoint representation of $\su(5)$ is given by
	\begin{equation*}\ytableausetup{boxsize=0.85em, aligntableaux=center}
	\ydiagram{2,1,1,1} \, .
	\end{equation*}
	
	\item The conjugate representation of a Young diagram whose first row contains $\ell$ boxes is given by the complementary diagram (rotated by 180 degrees) shaping the rectangle of $N$ rows and $\ell$ columns. For example, for $\su(5)$ the conjugate representation of
	\begin{equation*} \ytableausetup{boxsize=0.85em, aligntableaux=center}
	\ydiagram{3,2,2,1} \hspace{1cm}  \text{is}  \hspace{1cm}  \ydiagram{3,2,1,1}  \hspace{1cm} \text{since they build} \hspace{1cm} \ydiagram[*(white)]{3,2,2,1}*[*(gray)]{3,3,3,3,3}.
	\end{equation*}
	
	Note that the conjugate diagram of the fundamental representation is given by a single column of $N-1$ boxes, while the adjoint representation is self-conjugate.
\end{itemize}

Young diagrams are particularly helpful when decomposing the tensor product of two representations into a direct sum of irreducible representations.
Here, one follows two steps: first, one combines the boxes of the two diagrams by adding, one at a time, all boxes in the first row of the second diagram to the first one, respecting the condition of non-increasing length from top to bottom for the rows of the newly created diagrams and remembering that each of them can have at most $N$ rows. One repeats the procedure for all rows in the second diagram. As a second step, one discards all diagrams which do not satisfy specific rules that we are not going to mention here; for a full description, see ref.~\cite{Georgi}. Furthermore, for the algebra $\su(N)$ all columns with $N$ boxes occurring in a diagram can be deleted.

Recalling that the tensor product of the fundamental representation and its conjugate can be decomposed as a direct sum of the trivial and the adjoint representation and taking again $\su(5)$ as an example, we have
\begin{equation}\ytableausetup{boxsize=0.85em, aligntableaux=center}
\overline{U} \otimes U =\overline{\pi}_f  \otimes \pi_f
=
\quad  \ydiagram{1,1,1,1} \quad \otimes \quad \ydiagram{1} \quad
=
\quad \emptyset \quad \oplus \quad \ydiagram{2,1,1,1} \quad
=
\pi_1 \oplus \pi_{\mathrm{ad}} \, ,
\end{equation}
since a diagram with a column of $N=5$ boxes is equivalent to the trivial representation.

An alternative way to express an irreducible representation of $\su(n)$ is to associate a  \emph{Dynkin label} $\left(\lambda_1,\lambda_2,\dots,\lambda_{N-1}\right)$, where $\lambda_b$ gives the number of columns made of $b$ boxes. For instance, the fundamental representation is given by the label $(1,0,\dots,0)$ and the adjoint representation by  $(1,0,\dots,0,1)$.

\section{(Stochastic processes and Markov chains)}\label{app:stochastic_processes}

A (discrete) stochastic process with a sequence of random variables $X_1,X_2,\dots$ whose next step depends solely on the current state is called
a
\emph{Markov chain}. We consider a countable set of values $\Lambda=\left\{\lambda_1,\lambda_2,\dots\right\}$ which the variables $X_j$ can assume during the process and denote it as \emph{state space}.
For the variable $X_j$ we can then assign a \emph{probability distribution} $\omega_j=(\omega_j^1,\omega_j^2,\dots)$ where $\omega_j^k=\mathbb{P}(X_j=\lambda_k)$.
If the state space is finite, the transition from $X_j$ to $X_{j+1}$ can be described by a \emph{transition matrix} $P_j$ with entries
\begin{equation}
p_{a,b}=\mathbb{P}(X_{j+1}=\lambda_b | X_j=\lambda_a)
\end{equation}
such that we have
\begin{equation}
\omega_{j+1}=\omega_j\,P_j.
\end{equation}
If the process is homogeneous, then each transition is governed by the same transition matrix $P$, and
\begin{equation}
\omega_n=\omega_0\,P^n.
\end{equation}
The \emph{stationary distribution} of the process $\omega$ satisfies
\begin{equation}
\omega=\omega\,P
\end{equation}
and can hence be regarded as a fixed point of the chain.

For an ergodic chain, we refer as the \emph{mixing time} of the chain to the number of steps required to reach closeness to the stationary distribution.
For two arbitrary distributions $\omega$ and $\eta$, the \emph{total variation distance} is given by
\begin{equation}
\left\|\omega-\eta\right\|_{TV}=\frac{1}{2}\left\|\omega-\eta\right\|_1=\frac{1}{2}\sum_j\left|(\omega)_j-(\eta)_j\right| .
\end{equation}
Then the  \emph{mixing time} is defined as
\begin{equation}
\tau(\varepsilon) \coloneqq \max_{\omega_0}\min_{t\geq 0}\left\{t\   :  \left\|\omega_0 P^t-\omega \right\|_{TV}\leq \varepsilon\right\} \ ,
\end{equation}
where $\omega_0$ is the initial probability distribution and $\omega$ the stationary distribution.

The function $\left\{N(t)\ : \ t\geq 0\right\}$ counting the number of jumps occurred up to
the positive time $t$ defines a  \emph{Poisson process} if the following properties are satisfied.
\begin{enumerate}
	%[label=(\subscript{B}{\arabic*})]
	\item $N(0)=0$.
	\item The increments are independent and stationary.
	% 	\item For a small interval \mk{of size} $\Delta t$, $\mathbb{P}(N(\Delta t)=1)=\lambda\Delta t+o(\Delta t)$ and $\mathbb{P}(N(\Delta t)\geq 2)=o(\Delta t)$.
	\item Each increment $N(t+\Delta t) - N(t)$ is distributed as a Poisson random variable with parameter (mean) $\lambda t$.
\end{enumerate}
The last condition implies that $\EE[N(t)] = \lambda t$ and, in particular, the probability that two or more jumps occur in the time interval is negligible when it is small.

The waiting time $W$ between two consecutive jumps is then described by an  \emph{exponential distribution}, having
for $\lambda>0$ a  \emph{cumulative distribution function}
\begin{equation}
\mathbb{P}(W\leq t)=1-\e^{-\lambda\,t}
\end{equation}
and a  \emph{probability density function}
\begin{equation}
\mathit{f}(t)=\lambda\,\e^{-\lambda\,t}.
\end{equation}

\section{(Proof of Lemma~\ref{lemma:waiting_time})}\label{app:proof_Lemma_waiting_time}

		%[Proof of Lemma~\ref{lemma:waiting_time}]
		To prove this result on the waiting time, we first assume that we reach the region $\left[r_{-},r_{+}\right]$ within $S\leq s$ accelerated steps for some $s=\O(n)$ and we bound the probability that the waiting time exceeds $\varsigma\,n \log^2 n$. We will deal with the case of $\mathbb{P}(S>s)$ afterwards.
		Now, let $M$ be the smallest site visited during the walk, and let $\left\{y_i\right\}_{i=1}^S$ be a sequence of accelerated steps where $S \leq s$, with waiting times $\left\{W_i\right\}_{i=1}^S$ respectively, satisfying the event
		\begin{equation}
		H=\bigcap_{j=1}^n\left[\sum_{k=1}^S\mathbb{I}(X_k \leq j) \leq z j/\mu \right],
		\end{equation}
		where $\mathbb{I}$ is the indicator function and $X_k$ is the random variable assuming values in $\Omega=\{1,2,\dots,n\}$
		describing the state of the chain at step $k$ and $z$ chosen as $\O(\log n)$. In words, this means that, if $H$ occurs, then no site has been visited ``too often''. This is a useful event, since the smaller is the value of the current state of the chain, the smaller is the parameter of the exponential distribution dominating the waiting time. Namely, we have
		\begin{equation}
		1-P(k,k)=\frac{16k(3n-2k-1)}{n(n-1)} d t \geq \frac{16k}{n}\Delta t .
		\end{equation}
		So, dealing with three events, we consider the bound
		\begin{align}
		\mathbb{P}(W>t) &= \mathbb{P}(W>t\ \cap \ H\ \cap\ S\leq s)+\mathbb{P}(W>t\ \cap \ H\ \cap\ S> s)\nonumber \\
		&+\mathbb{P}(W>t\ \cap \ H^c\ \cap\ S\leq s)+\mathbb{P}(W>t\ \cap \ H^c\ \cap\ S> s) \nonumber\\
		& \leq \mathbb{P}(W>t\ | \ H\ \cap\ S\leq s)+\mathbb{P}(H\ \cap\ S> s) \nonumber\\
		&+\mathbb{P}(H^c\ \cap\ S\leq s)+\mathbb{P}(H^c\ \cap\ S> s) \nonumber\\
		&\leq \mathbb{P}(W>t\ | \ H\ \cap\ S\leq s) + \mathbb{P}(H^c\ |\ S \leq s)+\mathbb{P}( S > s) \label{eq:3_events}.
		\end{align}
		Conditioning on the two previous event and setting $M=m$ for arbitrary $m \in \left\{1,\dots,\ell \right\}$, we have to find an upper bound for the waiting time being too large; more precisely for a given run time $t$, we show:
		\begin{lemma}[Waiting time conditioning on event $H$]\label{lemma:conditioned_waiting_time}
			\begin{equation}
			\max_{\left\{y_i \right\}}\; \mathbb{P}\left( W(y_1)+\dots+W(y_S)\geq t\ \big| \ M=m\; ,\, H \right) \leq \e^{-\frac{8k}{n}t}\; 2^{z m / \mu} \e^{z m / (2 \mu)\, \log n}.
			\end{equation}
		\end{lemma}
		
		\begin{proof}  \emph{(Proof of Lemma \ref{lemma:conditioned_waiting_time})}
			We recall that this is the exactly the sequence visiting $m$ for $z m/ \mu$ (for simplicity, we assume it to be an integer) times and all other $j>m$ sites for $z / \mu$ times, hence
			\begin{equation}
			W(y_1)+\dots+W(y_S) \leq \sum_{i=1}^{z m/ \mu} E_{m,i}+\sum_{i=1}^{z/ \mu} \sum_{k=m+1}^r E_{k,i},
			\end{equation}
			where $E_{k,i}$ are i.i.d.\ exponential distributions with parameter $p(k)= {16k}/{n}$.
			Now applying Markov's inequality we obtain
			\begin{align}
			\mathbb{P}\left( \sum_{i=1}^{z m/ \mu} E_{m,i}+\sum_{i=1}^{z/ \mu} \sum_{k=m+1}^r E_{k,i} > t \ \right) & \leq \frac{ \E \left[ \exp \left\{\alpha \left(\sum_{i=1}^{z m/ \mu} E_{m,i}+\sum_{i=1}^{z/ \mu} \sum_{k=m+1}^r E_{k,i} \right)\right\} \right]  }{\e^{\alpha t}}\\
			&= \e^{-\alpha t} \; \left(\frac{p(m)}{p(m)-\alpha}\right)^{z m / \mu}\prod_{k=m+1}^r\left(\frac{p(k)}{p(k)-\alpha}\right)^{z/ \mu}
			\nonumber
			\end{align}
			for $\alpha<p(m)$. Let us choose $\alpha={p(m)}/{2}$, then we have
			\begin{align}
			\mathbb{P}\left( \sum_{i=1}^{z m/ \mu} E_{m,i}+\sum_{i=1}^{z/ \mu} \sum_{k=m+1}^r E_{k,i} > t \ \right) &\leq \e^{-\frac{8m}{n}t}\; 2^{z m / \mu} \left( \prod_{k=m+1}^{r} \frac{2k}{2k-m} \right) ^{z / \mu} \\
			&\leq \e^{-\frac{8m}{n}t}\; 2^{z m / \mu} \e^{z m / (2 \mu)\, \log n}.\nonumber
			\end{align}
			\proofend\end{proof}

		With this lemma we obtain an equivalent result for the waiting time as in ref.~\cite{Decoupling} up to the prefactor of $t$.
		Hence, for $t>\varsigma\,n \log^2 n$ with $\varsigma$ sufficiently large, applying the bounds on the probabilities $\mathbb{P}(M = m)$ for each value of $m \in \left\{1,\dots,\ell\right\}$ proved for the random quantum circuit case, we have
		\begin{align}
		\mathbb{P}\left(W_{r_{-}} >t\ | \ H\ \cap\ S\leq s\right)
		&=
		\sum_{m=1}^{\ell}\mathbb(M=m)\ \max_{\left\{y_i \right\}}\; \mathbb{P}\left( W(y_1)+\dots+W(y_S)\geq t\ | \ M=m \right) \\
		&\leq
		\frac{1}{(3-\eta)^\ell\binom{n}{\ell}}\frac{1}{{\rm poly}(n)}\nonumber.
		\end{align}
		The last two probability terms in eq.~\eqref{eq:3_events} depend only on the path of the accelerated random walk before reaching the interval $\left[r_{-},r_{+}\right]$.
		Looking at the accelerated chain and considering $\ell$ being in the region $\left[1,(3/4-\delta) n\right)$, we have
		${3(n-\ell)}/({3n-2l-1})\geq 1/2+\delta$ for any $n$. So, constructing a random walk $X'_k$ starting at the origin moving forward with probability $1/2+\delta$ and backward with $1/2-\delta$, it follows
		\begin{align}
		\mathbb{P}\left(S>s\right) & \leq \mathbb{P}\left(X'_s < r_{-}-\ell\right) \\
		&=\mathbb{P}\left(X'_s < 2\delta \, s-(2\delta \, s+\ell-r_{-})\right) \nonumber \\
		&\leq \exp\left ( -\frac{\left(2\delta s+\ell-r_{-}\right)^2}{2s}\right )\nonumber,
		\end{align}
		where in the last inequality we have used the Chernoff bound in ref.~\cite[Lemma A.3]{RandomCircuitsLow} assuming $2\delta s+\ell-r_{-}>0$. We conclude that the probability for the waiting time to be larger than $s\geq \phi n$ is exponentially decreasing in $n$ for large enough $\phi$.
		The last remaining term in eq.~\eqref{eq:3_events} can instead be bounded by (see ref.~\cite{Decoupling})
		\begin{equation}
		\mathbb{P}(H^c\ | \ S\leq s) \leq \frac{1}{(3-\eta)^\ell\binom{n}{\ell}}\frac{1}{{\rm poly}(n)}
		\end{equation}
		so that the proof of Lemma \ref{lemma:waiting_time} is now complete.
		\proofend

\end{appendices}

\end{document}